\documentclass[12pt,linesnumbered,ruled,vlined]{article}
\usepackage[latin9]{inputenc}
\usepackage{geometry}
\usepackage{xcolor}
\usepackage{verbatim}
\usepackage{float}
\usepackage{multirow}
\usepackage{algorithm2e}
\usepackage{amsmath}
\usepackage{amssymb}
\usepackage{graphicx}
\usepackage{setspace}
\usepackage[authoryear]{natbib}

\usepackage[normalem]{ulem} 
\usepackage[affil-it]{authblk} 
\newcommand\blfootnote[1]{%
  \begingroup
  \renewcommand\thefootnote{}\footnote{#1}%
  \addtocounter{footnote}{-1}%
  \endgroup
}
\makeatletter

\providecommand{\tabularnewline}{\\}

\usepackage{amsthm}
\usepackage{amsfonts}
\usepackage{latexsym}\usepackage{setspace}\usepackage{amsxtra}
\usepackage{bbm}
\usepackage{bm}
\usepackage{graphics}
\usepackage{multirow}
\usepackage{rotating, float}
\usepackage{xcolor}
\usepackage{booktabs}
\usepackage{lscape}

\newtheorem{definition}{Definition}

\newtheorem{model}{Model}
\newtheorem{lemma}{Lemma}
\newtheorem{corollary}{Corollary}
\newtheorem{proposition}{Proposition}
\newcommand{\E}{\mathbb{E}}
\newcommand{\Var}{var}
\newcommand{\cov}{cov}
\def \N{\mathbb{N}}
\def \R{\mathbb{R}}
\def \Ff{\mathcal{F}_t^{[full]}}
\def \Fn{\mathcal{F}_t^{[freq]}}
\def \Fa{\mathcal{F}_t^{[agg]}}
\def \P{\mathbb{P}}
\def \Z{\mathbf{Z}}

\title{Predictive Risk Analysis in Collective Risk Model: Choices between Historical Frequency and Aggregate Severity}

\author[1,*]{
Rosy Oh
}
\author[2,*]{
Youngju Lee
}
\author[3,**]{Dan Zhu
}
\author[2,**]{Jae Youn Ahn
}

\affil[1]{%
  Institute of Mathematical Sciences, Ewha Womans University, Korea.}
\affil[2]{%
  Department of Statistics, Ewha Womans University, Korea.}
\affil[3]{%
  Department of Econometrics and Business Statistics, Monash University, Australia}

\makeatother

\begin{document}
\maketitle
\blfootnote{*First Authors}
\blfootnote{**Corresponding Authors}
\blfootnote{Email addresses: \texttt{rosy.oh5@gmail.com} (Rosy Oh), \texttt{dldudwn220@gmail.com} (Youngju Lee),
\texttt{dan.zhu@monash.edu} (Dan Zhu), and \texttt{jaeyahn@ewha.ac.kr} (Jae Youn Ahn)}
\begin{abstract}
Typical risk classification procedure in insurance is consists of
a priori risk classification determined by observable risk characteristics,
and a posteriori risk classification where the premium is adjusted
to reflect the policyholder's claim history. While using the full
claim history data is optimal in a posteriori risk classification
procedure, i.e. giving premium estimators with the minimal variances,
some insurance sectors, however, only use partial information of the
claim history for determining the appropriate premium to charge. Classical
examples include that auto insurances premium are determined by the
claim frequency data and workers' compensation insurances are based
on the aggregate severity. The motivation for such practice is to
have a simplified and efficient posteriori risk classification procedure
which is customized to the involved insurance policy. This paper compares
the relative efficiency of the two simplified posteriori risk classifications,
i.e. based on frequency versus severity, and provides the mathematical
framework to assist practitioners in choosing the most appropriate
practice.
\end{abstract}

\section{Introduction}

Determination of the premiums is a major and interesting problem in
Actuarial science. Fair insurance premiums are established via risk
classification procedures, which involve the grouping of risks into
various classes that share a homogeneous set of characteristics allowing
the actuary to reasonably price discriminated. This paper examines
the statistical properties of the Collective Risk Model(CRM) for the
application in risk classification procedures. The CRM \citep{Klugman}
is, defined as the random sum, $S$, of claim severities,

\[
S=\sum_{j=1}^{N}Y_{j}
\]
where an independence assumption is made between the frequency, $N$,
and individual severities, $Y_{j}$'s. Due to the mathematical elegance
and relative robustness of the model, such independence assumption is widely applied by the
insurance industry for modeling the aggregate claim experience of
their portfolios over a fixed time horizon in pricing and reserving
exercises. However, in empirical studies such as \citep{Czado3},
the number and the size of claims are significantly dependent. Extending
from is its original form, there are various dependence structured
studied to better capture the stochastic nature of the insurance portfolio.

Generally speaking, there are two ways of describing CRM. The first
method is the \emph{two-part approach} where the frequency and the
severity part are described separately, and then their joint distribution
is described statistically via random effects \citep{Bastida, Czado2015, lu2016flexible, oh2019bonus},
copula specifications \citep{Czado, Kramer2013, Gee2016, Marceau2018, oh2019copula}, hierarchical structures
\citep{Peng, Garrido, park2018, AhnValdez2} or the dependence specification between the inter-arrival
time and the severity \citep{Albrecher2006, Boudreault2006, cheung2010structural}. Alternatively, the \textit{direct} \emph{approach}
models the distribution of aggregate severity, $S$, directly. The
most frequently used distribution for the aggregate severity in the
insurance literature is the Tweedie distribution \citep{tweedie1984index}.
To reflect the skewness and the heavy tails of the loss, \citep{mcdonald2008some}
introduced the generalized beta distribution of the second kind for
the modeling of the aggregate severity.

In this paper, we address the problem of risk classification procedure
in predicting the mean of aggregate severity, $S$, based on a set
of information. Usually, it involves first classifying risks based
on a priori risk classification procedure involving a priori risk characteristics,
i.e. risk characteristics of each policyholder at the moment of contract.
It forms the basis for premium settings when a policyholder is new
and insufficient information may be available. After a priori risk
classification, the policyholder is further classified based on the
claim history of each policyholder. This secondary classification
is called a posteriori risk classification. In effect, the resulting
a posteriori premium allows one to correct and adjust the previous a
priori premium making the price discrimination even fairer and more
reasonable.

While full information which consists of claim history of frequency and
severities is guaranteed to give the best posteriori prediction, we
are mainly interested in comparison of two simpler versions of the
posteriori risk classification methods
\begin{itemize}
\item one based on the historical frequency information only,
\item one based on the historical aggregate severity information only,
\end{itemize}
under the general dependence structure in frequencies and severities.
Such simplified a posteriori risk classification is common, and the
type of a posteriori risk classification depends on the characteristics
of insurance to facilitate efficient communication with the policyholder
while keeping the reasonable efficiency of a posteriori risk classification.
For example, auto insurance prefers to use the former type while the
workers\textquoteright{} compensation insurance generally adapts the
latter.

For the fair comparison of two a posteriori risk classification methods, we
need a CRM which can accommodate both a posteriori risk classification
methods. Since the historical frequency information cannot be used in the
direct approach, we use two-part approaches where both a posteriori
classification methods are possible. In particular, we use the two-step
frequency-severity model in \citet{AhnValdez2} where the analytical
comparison between two posteriori risk classifications is possible.

The actuarial credibility theory is designed for a posteriori rating
system that takes into account the history of claims as it emerges.
We construct and compare the quality of two simplified posteriori classification
methods, i.e. by conditioning on the aggregate claim history and the
claim frequency history respectively, via the B\"uhlmann estimators \citep{buhlmann2006course},
i.e. a linear version of a posteriori mean. B\"uhlmann estimators of
premium based on the history of frequency as well as severity are
first considered in \citet{hewitt1970credibility} and further studied
in \citet{frees2003multivariate} and \citet{goulet2006credibility}.
While those studies assumed independence between the frequency and
the severity, we develop B\"uhlmann premiums based on the historical frequency information as well as B\"uhlmann premiums based on the historical aggregate severity information
claim frequency under the dependence assumption between the frequency and
the severity.
Furthermore, we
derive the Mean square errors (MSE) of the B\"uhlmann premiums, which
facilities the analysis and comparison of the quality of two premiums. Our, non-technical,
yet equally important contribution to the insurance society is that,
with such analytical tools and related numerical study, we are the
first to provide a practical guideline for choosing the appropriate
posteriori risk classification method.

In the numerical study, we compare the quality of two B\"uhlmann premiums
under various scenarios, which hopefully provide the practical guideline
about the choice of premium method between the historical frequency information and
the historical aggregate severity information. In general, B\"uhlmann premiums based on
the historical aggregate severity information have a tendency to outperform B\"uhlmann premiums
based on the historical frequency information when there is a strong dependence among individual
severities, and vice versa. Yet, the preferable method also dynamically
changes over time as the number of observations increases. Hence,
there is no simple rule-of-thumb in deciding the appropriate information
set, but to make cases by case analysis.

We apply our analysis to a practical
application with real data obtained from the auto insurance of the
Wisconsin Local Government Property Insurance Fund (LGPIF) as in \citet{Frees4}.
First, the comparison of two B\"uhlmann premiums via numerical procedure
indicates relatively stronger dependence among severities, which in turn recommends the prediction based on the historical aggregate severity information rather than
the prediction based on the historical frequency information.
We confirm that the historical aggregate severity information has more predictive power in this particular example via out-of-sample validation.


The rest of the paper is organised as follows. In Section \ref{sec:Problem-Formulation},
we fix the notations and the model for our analysis, as well as present
a motivating problem. Section \ref{sec:Two-B=0000FChlmann-Premiums}
contains the derivation of the two B\"uhlmann premiums. In Section \ref{sec:Criteria},
we study the criteria for choosing appropriate models for premium
settings. An application to the auto insurance of the Wisconsin Local
Government Property Insurance Fund is presented in Section \ref{sec:Application}.

\section{Problem Formulation\label{sec:Problem-Formulation}}

\subsection{Notations}

We consider a portfolio of policyholders in the context of short-term
insurance, where a policyholder could decide whether or not to renew
the policy at the end of each policy year and the insurer can adjust
the premium at the beginning of each policy year based on the policyholder's
claim experience. We denote $\mathbb{N}$, $\mathbb{N}_{0}$, $\R$,
and $\R^{+}$ by the set of natural numbers, the set of non-negative
integers, the set of real numbers, and the set of positive real numbers,
respectively. We consider a set of discrete-time stochastic processes
that the associated data is collected over time $t=1,2,...$ that
\begin{itemize}
\item $N_{t}\in\N$ denotes the claim count at time $t$ and $\left\{ \Fn\right\} _{t=0}^{\infty}$
denotes the natural filtration generated by $N_{t}$,
\item $Y_{t,j}\in\R^{+}$ denote the size of the $i$th claim and
\[
\boldsymbol{{Y}}_{t}:=\begin{cases}
(Y_{t,1},...,Y_{t,N_{t}}) & N_{t}>0\\
undefined & N_{t}=0
\end{cases}
\]
denote the vector of claim sizes observed at time $t$,
\item $\left\{ \Ff\right\} _{t=0}^{\infty}$ denotes the natural filtration
generated by $N_{t}$ and $Y_{t}$, i.e. the information of the full
claim history
\item $S_{t}=\sum_{j=1}^{N_{t}}Y_{t,j}$ denotes the aggregate claim at
time $t$ and $\left\{ \Fa\right\} _{t=0}^{\infty}$ denotes the natural
filtration generated by $S_{t}$,
\item $M_{t}:=\begin{cases}
\frac{S_{t}}{N_{t}} & N_{t}>0\\
0 & N_{t}=0
\end{cases}$ denotes the average claim amount at time $t$.
\end{itemize}
We use lower case letters to denotes the realisation of these random
variables. The actuarial science literature often refers to $N_{t}$
as the \textit{frequency}, $\boldsymbol{{Y}}_{t}$ as the \textit{individual
severity}, and $M_{t}$ as the \textit{average severity} of the insurance
claims. We emphasize that our proposed method requires only information
on the average severity, and imposes no constraints on individual
severity. In the following text, we refer to the model for $(N_{t},M_{t})$
as the \textit{frequency-severity model}.

In the risk classification, premiums are determined by a set of risk
characteristics. Let $\left(\boldsymbol{X}^{[1]},\boldsymbol{X}^{[2]}\right)$
and $\left(R^{[1]},R^{[2]}\right)$ denote, respectively, the observed
and unobserved risk characteristics, and the superscripts {[}1{]}
and {[}2{]} are indices for the frequency and the severity components,
respectively. For convenience, we call the observed risk characteristics
$\left(\boldsymbol{X}^{[1]},\boldsymbol{X}^{[2]}\right)$ as a priori
risk characteristics. Note that $\left(\boldsymbol{X}^{[1]},\boldsymbol{X}^{[2]}\right)$,
and $\left(R^{[1]},R^{[2]}\right)$ do not have the subscript $t$,
as we assume that they are constant in time. The marginal distributions
of the residual effect characteristics are given by $R^{[1]}\sim G_{1}$
and $R^{[2]}\sim G_{2}$, respectively, for the proper distribution
functions $G_{1}$ and $G_{2}$. We use $g_{1}$ and $g_{2}$ to denote
the density version of $G_{1}$ and $G_{2}$, respectively. Furthermore,
denote
\[
\Z=(\boldsymbol{X}^{[1]},\boldsymbol{X}^{[2]},R^{[1]},R^{[2]}),\boldsymbol{\;\;X}:=\left(\boldsymbol{X}^{[1]},\boldsymbol{X}^{[2]}\right)\quad and\quad\boldsymbol{R}:=\left(R^{[1]},R^{[2]}\right)
\]
where $\Z$ is called as the risk characteristics.

\subsection{The Motivating problem}

Predictions of $S_{t+1}$ can be made based on different information
sets, i.e.

\begin{equation}
\E[S_{t+1}|\Ff]\label{eq:E_full}
\end{equation}
\begin{equation}
\E[S_{t+1}|\Fn]\label{eq:E_N}
\end{equation}
\begin{equation}
\E[S_{t+1}|\Fa]\label{eq:E_a}
\end{equation}
By definition, all three predictors are unbiased estimators of the
premium, and it is obvious that the quality of the predictor in Equation
\eqref{eq:E_full} is the best among three in the sense
\[
\Var\left(\E[S_{t+1}|\Ff]\right)\le\Var\left(\E[S_{t+1}|\Fn]\right)\quad and\quad\Var\left(\E[S_{t+1}|\Ff]\right)\le\Var\left(\E[S_{t+1}|\Fa]\right)
\]
which follows from the set inclusions
\[
\mathcal{F}_{t}^{[{\rm freq}]}\subseteq\mathcal{F}_{t}^{[{\rm full}]}\quad and\quad\mathcal{F}_{t}^{[{\rm agg}]}\subseteq\mathcal{F}_{t}^{[{\rm full}]}.
\]
In the comparison of the quality of two predictors in Equation \eqref{eq:E_N}
and \eqref{eq:E_a}, the statement that Equation \eqref{eq:E_a} is
superior in the sense
\[
\Var\left(\E[S_{t+1}|\Fa]\right)\le\Var\left(\E[S_{t+1}|\Fn]\right)
\]
does not hold in general. The main focus of this paper is to compare
these two approaches of setting premiums and set out guidelines in
choosing the appropriate when in practice.

\subsection{Model Assumption on Dependent Collective Risk Models\label{subsec:assump_CRM}}

To model the frequency and the severity of insurance claims, we follow
\citet{deJong2008} and use generalized linear models (GLMs). Specifically,
we consider the \textit{exponential dispersion family} (EDF) in \citet{Nelder1989}.
Now we are ready to present the collective risk model equipped with
various dependence structures.

\begin{model}\citep{AhnValdez2} \label{mod1} Suppose the insurer
predetermines $\mathcal{K}$ risk classes based on the policyholders'
risk characteristics. Let $\left(\boldsymbol{X}_{\kappa}^{[1]},\boldsymbol{X}_{\kappa}^{[2]}\right)$
define the a priori risk characteristics of the $\kappa$-th risk class, and $w_{\kappa}$ be the weight
of the risk class:
\[
w_{\kappa}:=\P\left(\boldsymbol{X}^{[1]}=\boldsymbol{x}_{\kappa}^{[1]},\boldsymbol{X}^{[2]}=\boldsymbol{x}_{\kappa}^{[2]}\right),\quad\kappa=1,\cdots,\mathcal{K}.
\]

Denote $\left(\Lambda^{[1]},\Lambda^{[2]}\right)$ as the \emph{a
priori premium for the given policyholder}, which is determined by
the observed risks characteristics as follows:
\begin{equation}
\Lambda^{[1]}=\left(\eta^{[1]}\right)^{-1}\left(\boldsymbol{X}^{[1]}\boldsymbol{\beta}^{[1]}\right)\quad and\quad\Lambda^{[2]}=\left(\eta^{[2]}\right)^{-1}\left(\boldsymbol{X}^{[2]}\boldsymbol{\beta}^{[2]}\right), \label{eq:prior}
\end{equation}
where $\eta^{[1]}(\cdot)$ and $\eta^{[2]}(\cdot)$ are link functions,
and $\boldsymbol{\beta}^{[1]}$ and $\boldsymbol{\beta}^{[2]}$ are
parameters to be estimated. We also assume that $R^{[1]},R^{[2]}$
and $\left(\boldsymbol{X}^{[1]},\boldsymbol{X}^{[2]}\right)$ are
independent. For the easiness of the analysis, we assume that priori
risk characteristics are fixed across time $t$. Assume that $N_{t}$
and $Y_{t,j}$'s are independent conditional on the risk
characteristics $\Z$.
\begin{itemize}
\item The frequency is specified using a count regression model conditioning
on the risk characteristics
\[
N_{t}\big\vert\Z\sim{F_{1}}\left(\cdot;\Lambda^{[1]}R^{[1]},\psi^{[1]}\right)
\]
for $t=1,2,...$ where the distribution $F_{1}$ has the mean $\Lambda^{[1]}R^{[1]}$
and the dispersion parameter $\psi^{[1]}$.
\item The individual severity $Y_{t,j}$ is specified using a regression
model conditioning on the risk characteristics, and the frequency
\begin{equation}
Y_{t,j}\big\vert\left(\Z,N_{t}\right)\sim F_{2}(\cdot;U^{[2]}R^{[2]},\psi^{[2]}),\quad i.i.d\;\;for\;\;j\in\mathbb{N},\label{eq.y}
\end{equation}
where the distribution $F_{2}$ has the mean $U^{[2]}R^{[2]}$ with
\[
U^{[2]}=\left(\eta^{[2]}\right)^{-1}\left(\boldsymbol{X}^{[2]}\boldsymbol{\beta}^{[2]}+\beta_{0}^{[2]}N_{t}\right)
\]
and the dispersion parameter $\psi^{[2]}$.
\end{itemize}
\end{model}

For the brevity of the notation, we only use log-link function for
$\eta^{[2]}$ such that
\[
U^{[2]}=\Lambda^{[2]}\exp\left(\beta_{0}^{[2]}N_{t}\right),
\]
where
\[
\Lambda^{[2]}:=\exp\left(\boldsymbol{X}^{[2]}\boldsymbol{\beta}^{[2]}\right).
\]
To further simplify the analysis in the later sections, we shall assume
the following parametric model.

\begin{model}[Parametric Model]\label{mod2} Assume the settings
in Model \ref{mod1} with the following parametric assumptions:
\begin{enumerate}
\item For the frequency, assume that
\[
N_{t}\big\vert\Z\sim{\rm Pois}\left(\Lambda^{[1]}R^{[1]}\right).
\]
\item For the individual severity, assume that
\begin{equation}
Y_{t,j}\big\vert\left(\Z,N_{t}\right)\sim{\rm Gamma}\left(\Lambda^{[2]}\exp\left(\beta_{0}^{[2]}N_{t}\right)R^{[2]},\psi^{[2]}\right),\quad t,j\in\mathbb{N}.\label{eq.y2}
\end{equation}
\item For the random effect, assume that
\[
\E\left[{R^{[1]}}\right]=\E\left[{R^{[2]}}\right]=1
\]
and
\[
\Var\left[{R^{[1]}}\right]=b^{[1]}\quad and\quad\Var{\left[R^{[2]}\right]}=b^{[2]}.
\]
\end{enumerate}
\end{model}

We provide the mean, variance and covariance formulae for the statistics
of Model \ref{mod2} in Proposition \ref{prop.1} (see Appendix). It
is often the case that using the individual severity in \eqref{eq.y2}
is inconvenient for estimation purposes, yet insurance literature
often provides the distributional assumption for the average severity
\citep{Peng}. The following corollary provides the equivalence of
two representations based on the individual severity and the average
severity in the case of gamma distributional assumption. The same
result for EDF distribution can be found in Lemma \ref{App:lem1} (see
Appendix).

\begin{corollary} Consider the settings in Model \ref{mod1}, then
the individual severity assumption in \eqref{eq.y2} is the equivalent
with
\[
M_{t}\big\vert\left(\Z,N_{t}\right)\sim{\rm Gamma}\left(\Lambda^{[2]}\exp\left(\beta_{0}^{[2]}N_{t}\right)R^{[2]},\psi^{[2]}\right).
\]
for $N_{t}>0$, and
\[
\P\left(M_{t}=0\big\vert\Z,N_{t}\right)=1
\]
for $N_{t}=0$. \end{corollary}

\begin{proof}
The proof is an immediate result of Lemma \ref{App:lem1}.
\end{proof}

Finally, for the brevity of the paper, define the following symbols
under the settings in Model \ref{mod2}.

\begin{definition}\label{sym.1} Under the settings in Model \ref{mod2},
define
\[
\zeta_{1}:=\Lambda^{[1]}\left(e^{\beta_{0}^{[2]}}-1\right)\quad and\quad \zeta_{2}:=\Lambda^{[1]}\left(e^{2\beta_{0}^{[2]}}-1\right).
\]

\end{definition}

For the brevity of the paper, we also abuse symbols $\Lambda^{[1]}$,
$\Lambda^{[2]}$, and $\boldsymbol{\Lambda}=(\Lambda^{[1]},\Lambda^{[2]})$
as follows in a clear context. Under the settings in Model \ref{mod2},
we use $\Lambda^{[1]}$, $\Lambda^{[2]}$, and $\boldsymbol{\Lambda}$
to stand for $\boldsymbol{X}^{[1]}$, $\boldsymbol{X}^{[2]}$, and
$\boldsymbol{X}$, respectively, in a clear context. For example,
we have the following two expressions are equivalent
\[
\E\left[{S_{t+1}|\boldsymbol{\Lambda}}\right]\quad \hbox{and} \quad\E\left[{S_{t+1}|\boldsymbol{X}}\right].
\]

\section{Two B\"uhlmann Premiums in Dependent CRM}

\label{sec:Two-B=0000FChlmann-Premiums}

With the introduction of more complicated and dynamic insurance products,
a major challenge of the actuarial profession can be found in the
measurement and construction of a fair insurance premium. Pricing
risks based upon certain specific characteristics has a long history
in actuarial science. In light of the heterogeneity within an insurance
portfolio, an insurance company should not apply the same premium
for its policyholders, but group the risks in the portfolio so that
people with similar risk profiles pay the same reasonable premium
rate. To reflect the various risk profiles in a portfolio within a
statistically sound basis, the standard technique that actuaries use
is a regression-based approach. The standard GLM type structure as
in Equation \eqref{eq:prior} gives a natural candidate for the a priori risk
classification.

Based on such a priori risk classification, this section considers two different type of premiums under the dependent assumption between frequency and severity.
The first premium is a classical method where
the historical aggregate severity information is used to predict the aggregate severity
in the future.
The second premium that we are proposing is
a non-traditional approach in that the historical frequency information
is used to predict the aggregate severity in the future.
For the analytical of comparison of two methods,
we consider B\"uhlmann premiums rather using the exact posteriori mean of the aggregate severity as a posteriori premiums.
Note that B\"uhlmann premiums can deal with both approaches as long as one can
calculate the covariance matrix of the observations and premium \citep{hewitt1970credibility}.

%
%

\subsection{B\"uhlmann Premiums based on the Historical Aggregate Severity Information}

Under Model \ref{mod2}, our goal in this subsection is to find the
B\"uhlmann premium based on the historical aggregate severity.

\begin{definition}\label{def:P_agg}

The B\"uhlmann premium based on the historical aggregate severity information is
\[
{\rm Prem}_{1}(\boldsymbol{\Lambda}):=\widehat{\alpha}_{0}+\widehat{\alpha}_{1}S_{1}+\cdots+\widehat{\alpha}_{t}S_{t},
\]
where
\[
\left(\widehat{\alpha}_{0},\cdots,\widehat{\alpha}_{t}\right):=\arg\min\limits _{(\alpha_{0},...,\alpha_{t})\in\R^{t+1}}\E\left[\Big(\E\left[S_{t+1}|\boldsymbol{R},\boldsymbol{\Lambda}\right]-\left(\alpha_{0}+\alpha_{1}S_{1}+\cdots+\alpha_{t}S_{t}\right)\Big)^2\Big\vert\boldsymbol{\Lambda}\right].
\]

\end{definition}

For the known random effect $\boldsymbol{R}$ and a priori rate $\boldsymbol{\Lambda}$,
the conditional mean $\E{\left[S_{t+1}|\boldsymbol{R},\boldsymbol{\Lambda}\right]}$
can be obtained as follow.

\begin{proposition} \label{Prop_P_agg_1} Under Model \ref{mod2},
for the known random effect $\boldsymbol{R}$ and a priori rate $\boldsymbol{\Lambda}$,
we have
\[
\E\left[{S_{t+1}|\boldsymbol{R},\boldsymbol{\Lambda}}\right]=\Lambda^{[1]}\Lambda^{[2]}R^{[1]}R^{[2]}e^{\beta_{0}^{[2]}}\exp\left(\Lambda^{[1]}R^{[1]}\left(e^{\beta_{0}^{[2]}}-1\right)\right).
\]
\end{proposition}

\begin{proof}
\[
\begin{aligned}\E\left[{S_{t+1}|\boldsymbol{R},\boldsymbol{\Lambda}}\right] & =\E\left[{\E{\left[S_{t+1}|N_{t+1},\boldsymbol{R},\boldsymbol{\Lambda}\right]}\big\vert\boldsymbol{R},\boldsymbol{\Lambda}}\right]\\
 & =\E\left[{N_{t+1}R^{[2]}\Lambda^{[2]}\exp\left(\beta_{0}^{[2]}N_{t+1}\right)\big\vert\boldsymbol{R},\boldsymbol{\Lambda}}\right]\\
 & =R^{[2]}\Lambda^{[2]}\E\left[{N_{t+1}\exp\left(\beta_{0}^{[2]}N_{t+1}\right)\big\vert\boldsymbol{R},\boldsymbol{\Lambda}}\right]\\
 & =\Lambda^{[1]}\Lambda^{[2]}R^{[1]}R^{[2]}e^{\beta_{0}^{[2]}}\exp\left(\Lambda^{[1]}R^{[1]}\left(e^{\beta_{0}^{[2]}}-1\right)\right),
\end{aligned}
\]
where the last equality comes from Lemma \ref{dan.lem.1} in Appendix.
\end{proof}

Note that the B\"uhlmann premium in \eqref{def:P_agg} can be regarded
as the best linear unbiased estimator (BLUE) of the conditional mean
$\E{\left[S_{t+1}|\Fa,\boldsymbol{\Lambda}\right]}$. By the classical
procedure in \citet{buhlmann2006course}, one can easily show that
\begin{equation}
\widehat{\alpha}_{0}=(1-Z_{1}(\boldsymbol{\Lambda}))\E{\left[S_{t}|\boldsymbol{\Lambda}\right]}\quad and\quad\widehat{\alpha}_{1}=\cdots=\widehat{\alpha}_{t}=\frac{Z_{1}(\boldsymbol{\Lambda})}{t},  \label{alpha_agg}
\end{equation}
where B\"uhlmann factor is given by
\[
Z_{1}(\boldsymbol{\Lambda})=\frac{t\,\Var\left[\E\left[{S_{t}|\boldsymbol{\Lambda},\boldsymbol{R}}|\boldsymbol{\Lambda}\right]\right]}{\E\left[\Var\left[S_{t}|\boldsymbol{\Lambda},\boldsymbol{R}\right]|\boldsymbol{\Lambda}\right]+t\,\Var\left[\E\left[S_{t}|\boldsymbol{\Lambda},\boldsymbol{R}\right]|\boldsymbol{\Lambda}\right]}
\]

The
following result provides the analytical expression of premium in
\eqref{def:P_agg}.

\begin{proposition}\label{prop:P_agg_2} Under Model \ref{mod2},
the conditional mean $\E\left[S_{t}|\boldsymbol{\Lambda}\right]$
and B\"uhlmann factor, $Z_{1}(\Lambda)$, can be expressed as
\[
\E\left[S_{t}|\boldsymbol{\Lambda}\right]=u_{1}\left(\boldsymbol{\Lambda}\right)\quad and\quad Z_{1}(\boldsymbol{\Lambda})=\frac{ta_{1}\left(\boldsymbol{\Lambda}\right)}{ta_{1}\left(\boldsymbol{\Lambda}\right)+v_{1}\left(\boldsymbol{\Lambda}\right)},
\]
where
\[
u_{1}\left(\boldsymbol{\Lambda}\right)
=\Lambda^{[1]}\Lambda^{[2]}e^{\beta_{0}^{[2]}}\,M_{R^{[1]}}^{\prime}\left(\zeta_{1}\right)
\]

\[
v_{1}\left(\boldsymbol{\Lambda}\right)%
=\Lambda^{[1]}\left(\Lambda^{[2]}\right)^{2}e^{2\beta_{0}^{[2]}}\left(1+b^{[2]}\right)\bigg[\left(1+\psi^{[2]}\right)M_{R^{[1]}}^{\prime}\left(\zeta_{2}\right)+\Lambda^{[1]}e^{2\beta_{0}^{[2]}}M_{R^{[1]}}^{\prime\prime}\left(\zeta_{2}\right)-\Lambda^{[1]}M_{R^{[1]}}^{\prime\prime}\left(2\zeta_{1}\right)\bigg],
\]
and
\[
a_{1}\left(\boldsymbol{\Lambda}\right)%
=\left(\Lambda^{[1]}\Lambda^{[2]}\right)^{2}e^{2\beta_{0}^{[2]}}\bigg[\left(1+b^{[2]}\right)M_{R^{[1]}}^{\prime\prime}\left(2\zeta_{1}\right)-\left\{ M_{R^{[1]}}^{\prime}\left(\zeta_{1}\right)\right\} ^{2}\bigg].
\]

\end{proposition}

\begin{proof}
Proof of $u_{1}$ and $v_{1}$ immediately follow from Lemma \ref{dan.lem.st}.
Proof of $a_{1}$ is from
\[
\begin{aligned}a_{1}\left(\boldsymbol{\Lambda}\right) & =\Var{\left[u_{1}\left(\boldsymbol{R},\boldsymbol{\Lambda}\right)|\boldsymbol{\Lambda}\right]}\\
 & =e^{2\beta_{0}^{[2]}}\left(\Lambda^{[1]}\Lambda^{[2]}\right)^{2}\Var{\left[R^{[1]}R^{[2]}\exp\left(\Lambda^{[1]}R^{[1]}\left(e^{\beta_{0}^{[2]}}-1\right)\right)|\boldsymbol{\Lambda}\right]}\\
 & =e^{2\beta_{0}^{[2]}}\left(\Lambda^{[1]}\Lambda^{[2]}\right)^{2}\Bigg[\E{\left[\left(R^{[1]}R^{[2]}\right)^{2}\exp\left(2\Lambda^{[1]}R^{[1]}\left(e^{\beta_{0}^{[2]}}-1\right)\right)|\boldsymbol{\Lambda}\right]}\\
 & \quad\quad\quad\quad\quad\quad\quad\quad\quad\quad\quad\quad-\left(\E{\left[R^{[1]}R^{[2]}\exp\left(\Lambda^{[1]}R^{[1]}\left(e^{\beta_{0}^{[2]}}-1\right)\right)|\boldsymbol{\Lambda}\right]}\right)^{2}\Bigg]\\
 & =\left(\Lambda^{[1]}\Lambda^{[2]}\right)^{2}e^{2\beta_{0}^{[2]}}\bigg[\left(1+b^{[2]}\right)M_{R^{[1]}}^{\prime\prime}\left(2\zeta_{1}\right)-\left\{ M_{R^{[1]}}^{\prime}\left(\zeta_{1}\right)\right\} ^{2}\bigg].
\end{aligned}
\]
\end{proof}

The statistics in Proposition \ref{prop:P_agg_2} can be further explicitly
calculated based on the formulas for the moment generating function
of random effect $R^{[1]}$ in Lemma \ref{dan.lem.2} in Appendix.
Hence, we have

\begin{equation}
{\rm Prem}_{1}(\boldsymbol{\Lambda})=Z_{1}(\boldsymbol{\Lambda})\frac{\sum\limits _{k=1}^{t}S_{k}}{t}+\left(1-Z_{1}(\boldsymbol{\Lambda})\right)u_{1}\left(\boldsymbol{\Lambda}\right),  \label{eq:P_agg}
\end{equation}
where B\"uhlmann factors $Z_{1}(\boldsymbol{\Lambda})$ are described
in Proposition \ref{prop:P_agg_2} and \ref{Prop:P_freq_1}.

\subsection{B\"uhlmann Premiums based on the Historical Frequency Information}

Similar to the previous subsection, our goal in this subsection is
to derive the B\"uhlmann premium based on claim frequencies. The followings
\begin{equation}
\E\left[{S_{1}|N_{1},\boldsymbol{\Lambda}}\right],\cdots,\E\left[{S_{t}|N_{t},\boldsymbol{\Lambda}}\right]\label{eq:Buhlman_observation}
\end{equation}
are the \textit{B\"uhlmann observation}, that are functions of the historical
frequency and priori characteristics, are used as if they are observations
in the determination of the B\"uhlmann premiums, and we shall denote
\[
\widetilde{S}_{t}\left(N_{t},\boldsymbol{\Lambda}\right):=\E\left[S_{t}|N_{t},\boldsymbol{\Lambda}\right].
\]
\textit{ }In Model \ref{mod2}, the statistics in \eqref{eq:Buhlman_observation}
can be analytically expressed as follows.

\begin{proposition}\label{Prop:P_freq_1} Under the settings in Model
\ref{mod2}, the B\"uhlmann observation $\widetilde{S}_{t}\left(N_{t},\boldsymbol{\Lambda}\right)$
in \eqref{def:P_freq} is expressed as
\[
\widetilde{S}_{t}\left(N_{t},\boldsymbol{\Lambda}\right)=\Lambda^{[2]}N_{t}\exp\left(\beta_{0}^{[2]}N_{t}\right).
\]
\end{proposition}

\begin{proof} Based on the law of total expectation, we have
\[
\begin{aligned}\widetilde{S}_{t}\left(N_{t},\boldsymbol{\Lambda}\right) & =\E\left[S_{t}|N_{t},\boldsymbol{\Lambda}\right]\\
 & =\E\left[\E\left[S_{t}|N_{t},\boldsymbol{\Lambda},\boldsymbol{R}\right]|N_{t},\boldsymbol{\Lambda}\right]\\
 & =\E\left[N_{t}\E\left[M_{t}|N_{t},\boldsymbol{\Lambda},\boldsymbol{R}\right]|N_{t},\boldsymbol{\Lambda}\right]\\
 & =\E\left[{N\Lambda^{[2]}R^{[2]}\exp\left(\beta_{0}^{[2]}N_{t}\right)|N_{t},\boldsymbol{\Lambda}}\right]\\
 & =\Lambda^{[2]}N_{t}\exp\left(\beta_{0}^{[2]}N_{t}\right).
\end{aligned}
\]
\end{proof}

\begin{definition}\label{def:P_freq}

The B\"uhlmann premium based on the historical frequency information is defined
as

\[
{\rm Prem}_{2}(\boldsymbol{\Lambda}):=\widehat{\alpha}_{0}^{*}+\widehat{\alpha}_{1}^{*}\widetilde{S}_{1}\left(N_{1},\boldsymbol{\Lambda}\right)+\cdots+\widehat{\alpha}_{t}^{*}\widetilde{S}_{t}\left(N_{t},\boldsymbol{\Lambda}\right),
\]
where
\[
\left(\widehat{\alpha}_{0},\cdots,\widehat{\alpha}_{t}\right):=\arg\min\limits _{(\alpha_{0},\cdots,\alpha_{t})\in\R^{t+1}}\E\left[\left(\E\left[S_{t+1}|\boldsymbol{R},\boldsymbol{\Lambda}\right]-\left(\alpha_{0}+\alpha_{1}\widetilde{S}_{1}\left(N_{1},\boldsymbol{\Lambda}\right)+\cdots+\alpha_{t}\widetilde{S}_{t}\left(N_{t},\boldsymbol{\Lambda}\right)\right)\right)^{2}\bigg\vert\boldsymbol{\Lambda}\right].
\]

\end{definition}

Similar to B\"uhlmann type premium in \eqref{def:P_agg},
the classical procedure in \citet{buhlmann2006course} shows that
\begin{equation}
\widehat{\alpha}_{0}^{*}=(1-Z_{2}(\boldsymbol{\Lambda}))\E\left[{\widetilde{S}_{t}\left(N_{t},\boldsymbol{\Lambda}\right)|\boldsymbol{\Lambda}}\right]\quad and\quad\widehat{\alpha}_{1}^{*}=\cdots=\widehat{\alpha}_{t}^{*}=\frac{Z_{2}(\boldsymbol{\Lambda})}{t},   \label{alpha_freq}
\end{equation}
where the B\"uhlmann factor is
\[
Z_{2}(\boldsymbol{\Lambda}):=\frac{t\,\Var\left[{\E\left[\widetilde{S}_{t}\left(N_{t},\boldsymbol{\Lambda}\right)|\boldsymbol{\Lambda},\boldsymbol{R}\right]|\boldsymbol{\Lambda}}\right]}{\E\left[\Var\left[\widetilde{S}_{t}\left(N_{t},\boldsymbol{\Lambda}\right)|\boldsymbol{\Lambda},\boldsymbol{R}\right]|\boldsymbol{\Lambda}\right]+t\,\Var\left[\E\left[\widetilde{S}_{t}\left(N_{t},\boldsymbol{\Lambda}\right)|\boldsymbol{\Lambda},\boldsymbol{R}\right]|\boldsymbol{\Lambda}\right]}.
\]
The following result provides the analytical expression of premium
in \eqref{def:P_freq}.

\begin{proposition}\label{Prop:P_freq_2} Under Model \ref{mod2},
the conditional mean $\E{\left[\widetilde{S}_{t}\left(N_{t},\boldsymbol{\Lambda}\right)|\boldsymbol{\Lambda}\right]}$
and B\"uhlmann factor, $Z_{2}(\Lambda)$, can be expressed as
\[
\E{\left[\widetilde{S}_{t}\left(N_{t},\boldsymbol{\Lambda}\right)|\boldsymbol{\Lambda}\right]}=u_{2}\left(\boldsymbol{\Lambda}\right)\quad and\quad Z_{2}(\boldsymbol{\Lambda})=\frac{ta_{2}\left(\boldsymbol{\Lambda}\right)}{ta_{2}\left(\boldsymbol{\Lambda}\right)+v_{2}\left(\boldsymbol{\Lambda}\right)},
\]
where
\[
u_{2}\left(\boldsymbol{\Lambda}\right)=e^{\boldsymbol{\beta}_{0}^{[2]}}\Lambda^{[1]}\Lambda^{[2]}M_{R^{[1]}}^{\prime}\left(\zeta_{1}\right)
\]

\[
v_{2}\left(\boldsymbol{\Lambda}\right)  =
\Lambda^{[1]}\left(\Lambda^{[2]}\right)^{2}e^{2\beta_{0}^{[2]}} \left[
\Lambda^{[1]}e^{2\beta_{0}^{[2]}}M_{R^{[1]}}^{\prime\prime}\left(\zeta_{2}\right)+
M_{R^{[1]}}^{\prime}\left(\zeta_{2}\right)-
\Lambda^{[1]}M_{R^{[1]}}^{\prime\prime}\left(2\zeta_{1}\right)
\right]
\]
and
\[
a_{2}\left(\boldsymbol{\Lambda}\right)=e^{2\beta_{0}^{[2]}}\left(\Lambda^{[1]}\Lambda^{[2]}\right)^{2}\left[M_{R^{[1]}}^{\prime\prime}\left(2\zeta_{1}\right)-\left(M_{R^{[1]}}^{\prime}\left(\zeta_{1}\right)\right)^{2}\right].
\]

\end{proposition}

\begin{proof}
Proofs of $u_{2}$ and $v_{2}$ immediately follow from Lemma
\ref{dan.lem.3}. Proof of $a_{2}$ is from
\[
\begin{aligned}a_{2}\left(\boldsymbol{\Lambda}\right) & =\Var{\left[u_{2}\left(\boldsymbol{R},\boldsymbol{\Lambda}\right)|\boldsymbol{\Lambda}\right]}\\
 & =e^{2\beta_{0}^{[2]}}\left(\Lambda^{[1]}\Lambda^{[2]}\right)^{2}\Var{\left[R^{[1]}\exp\left(\Lambda^{[1]}R^{[1]}\left(e^{\beta_{0}^{[2]}}-1\right)\right)|\boldsymbol{\Lambda}\right]}\\
 & =e^{2\beta_{0}^{[2]}}\left(\Lambda^{[1]}\Lambda^{[2]}\right)^{2}\bigg(\E{\left[\left(R^{[1]}\right)^{2}\exp\left(2\Lambda^{[1]}R^{[1]}\left(e^{\beta_{0}^{[2]}}-1\right)\right)|\boldsymbol{\Lambda}\right]}\\
 & \quad\quad\quad\quad\quad\quad\quad\quad\quad\quad\quad\quad-\left(\E{\left[R^{[1]}\exp\left(\Lambda^{[1]}R^{[1]}\left(e^{\beta_{0}^{[2]}}-1\right)\right)|\boldsymbol{\Lambda}\right]}\right)^{2}\bigg)\\
 & =e^{2\beta_{0}^{[2]}}\left(\Lambda^{[1]}\Lambda^{[2]}\right)^{2}\left[M_{R^{[1]}}^{\prime\prime}\left(2\zeta_{1}\right)-\left(M_{R^{[1]}}^{\prime}\left(\zeta_{1}\right)\right)^{2}\right].
\end{aligned}
\]
\end{proof}

Note that the statistics in Proposition \ref{Prop:P_freq_2} can be
further explicitly calculated based on the expression for the moment
generating function of $R^{[1]}$ in Lemma \ref{dan.lem.2}. Hence,
we have
\begin{equation}
{\rm Prem}_{2}(\boldsymbol{\Lambda})=Z_{2}(\boldsymbol{\Lambda})\frac{\sum\limits _{k=1}^{t}\widetilde{S}_{k}\left(N_{k},\boldsymbol{\Lambda}\right)}{t}+\left(1-Z_{2}(\boldsymbol{\Lambda})\right)u_{2}\left(\boldsymbol{\Lambda}\right), \label{oh.eq.2}
\end{equation}
where B\"uhlmann factor  $Z_{2}(\boldsymbol{\Lambda})$
is described in Proposition \ref{prop:P_agg_2} and \ref{Prop:P_freq_1}.

\subsection{Linkage with B\"uhlmann Premium for the Frequency}

This section assumes the independence between frequency and individual
severities by assuming $\beta_{0}^{[2]}=0$ in Model \ref{mod2}. From the
assumption $\beta_{0}^{[2]}=0$ in Model \ref{mod2}, we have the following
intuitive interpretation about the B\"uhlmann observation and B\"uhlmann
premium based on the historical frequency information
\[
\widetilde{S}_{t}\left(N_{t},\boldsymbol{\Lambda}\right)=N_{t}\Lambda^{[2]}
\]
and
\begin{equation}
{\rm Prem}_{2}(\boldsymbol{\Lambda})=\Lambda^{[2]}\left(Z_{2}^{*}(\boldsymbol{\Lambda})\frac{\sum\limits_{k=1}^t N_{k}}{t}+(1-Z_{2}^{*}(\boldsymbol{\Lambda}))\Lambda^{[1]}\right),  \label{ahn.12}
\end{equation}
where
\[
Z_{2}^{*}(\boldsymbol{\Lambda})=\frac{ta_{2}^{*}\left(\boldsymbol{\Lambda}\right)}{ta_{2}^{*}\left(\boldsymbol{\Lambda}\right)+v_{2}^{*}\left(\boldsymbol{\Lambda}\right)}
\]
with
\[
v_{2}^{*}\left(\boldsymbol{\Lambda}\right):=\E{\left[\Var{\left[N_{t}|\boldsymbol{R},\boldsymbol{\Lambda}\right]}|\boldsymbol{\Lambda}\right]}\quad and\quad a_{2}^{*}\left(\boldsymbol{\Lambda}\right):=\Var{\left[\E{\left[N_{t}|\boldsymbol{R},\boldsymbol{\Lambda}\right]}|\boldsymbol{\Lambda}\right]}.
\]
Note that the expression
\[
Z_{2}^{*}(\boldsymbol{\Lambda})\frac{\sum\limits_{k=1}^t N_{k}}{t}+(1-Z_{2}^{*}(\boldsymbol{\Lambda}))\Lambda^{[1]}
\]
in \eqref{ahn.12} coincides with the B\"uhlmann premium of frequency
defined by
\begin{equation}
{\rm Prem}_{3}(\boldsymbol{\Lambda}):=\widehat{\alpha}_{0}+\widehat{\alpha}_{1}N_{1}+\cdots+\widehat{\alpha}_{t}N_{t}\label{ahn.11}
\end{equation}
where
\[
\left(\widehat{\alpha}_{0},\cdots,\widehat{\alpha}_{t}\right):=\arg\min\limits _{(\alpha_{0},\cdots,\alpha_{t})\in\R^{t+1}}\E{\left[\Big(\E{\left[N_{t+1}|\boldsymbol{R},\boldsymbol{\Lambda}\right]}-\left(\alpha_{0}+\alpha_{1}N_{1}+\cdots+\alpha_{t}N_{t}\right)\Big)^{2}\Big\vert\boldsymbol{\Lambda}\right]}.
\]

\section{Numerical comparisons of the two B\"uhlmann premiums\label{sec:Criteria}}

The aggregate claim amount is the key element for an insurer\textquoteright s
balance sheet, as it represents the amount of money paid on claims,
hence they must understand the dynamics of the aggregate claim overtime.
Yet, depending on features of contracts as well as policyholder behaviour
and risk mitigation practices, some insurance products use only partial
information on the claim history for the posteriori risk classification,
i.e. the frequency of claims or the aggregate claim amounts, but not
both. We conduct a numerical study to investigate the effect of various
dependence structures in Model \ref{mod2} on the two proposed B\"uhlmann
premiums, ${\rm Prem}_{1}(\boldsymbol{\Lambda})$ and ${\rm Prem}_{2}(\boldsymbol{\Lambda})$
in Section \ref{sec:Two-B=0000FChlmann-Premiums} that based on the
frequency and the aggregate claim respectively. The following analysis
forms the basis of choosing the appropriate pricing and risk mitigation
practices for standard general insurance portfolios.

\subsection{Numerical Set-up}

We assume only one priori risk class for the simplicity of the analysis
and the following parametric assumption for the unobserved heterogeneities
\[
\begin{cases}
R^{[1]}\sim & IG(1,b^{[1]});\\
R^{[2]}~\sim & {\rm Gamma}(1,b^{[2]}),
\end{cases}
\]
where IG is inverse Gaussian distribution.
In particular, we set
\[
\lambda^{[1]}=\lambda_{0}^{[1]}\quad\hbox{and}\quad\lambda^{[2]}=\lambda_{0}^{[2]}
\]
with $\lambda_{0}^{[1]}=\exp(-1.9)$ and $\lambda_{0}^{[2]}=\exp(8.4)$.
The parameters $(b^{[1]},b^{[2]},\beta_{0}^{[2]})$ varies in 27 scenarios
that different dependence structure is considered. Note here, $\beta_{0}^{[2]}$
controls the dependence between the frequency and severity whereas
$b^{[1]}$ and $b^{[2]}$ controls the correlation among the frequencies
and severities over time, respectively. For each scenario with the
combination of parameters $\left(b^{[1]},b^{[2]},\beta_{0}^{[2]}\right)$,
we choose
\[
\psi_{0}^{[2]}=\left\{ \dfrac{c}{(\lambda^{[2]})^{2}}+(M_{R^{[1]}}(\zeta_{1}))^{2}\right\} \dfrac{1}{(1+b^{[2]})M_{R^{[1]}}(\zeta_{2})}-1,
\]
so that $c:=\Var{\left[Y_{t,j}|\boldsymbol{\Lambda}=\left(\lambda_{0}^{[1]},\lambda_{0}^{[2]}\right)\right]}=2.008$
is fixed.

For the comparison of the two B\"uhlmann premiums, define the conditional
mean square errors as
\[
{\rm HMSE}_{1}\left(\boldsymbol{\Lambda},t\right):=\E{\left[\left(\E{\left[S_{t+1}|\boldsymbol{R},\boldsymbol{\Lambda}\right]}-{\rm Prem}_{1}\left(\boldsymbol{\Lambda},\mathcal{F}_{t}^{[{\rm agg}]}\right)\right)^{2}|\boldsymbol{\Lambda}\right]}
\]
and
\[
{\rm HMSE}_{2}\left(\boldsymbol{\Lambda},t\right):=\E{\left[\left(\E{\left[S_{t+1}|\boldsymbol{R},\boldsymbol{\Lambda}\right]}-{\rm Prem}_{2}\left(\boldsymbol{\Lambda},\mathcal{F}_{t}^{[{\rm freq}]}\right)\right)^{2}|\boldsymbol{\Lambda}\right]}
\]
for the two B\"uhlmann premiums obtained under Model \ref{mod2}. Using
the above definition, we have
\begin{equation}
{\rm HMSE}_{1}\left(t\right):=\sum\limits _{\kappa\in\mathcal{K}}w_{\kappa}\,{\rm HMSE}_{1}\left(\boldsymbol{\lambda}_{\kappa},t\right)\quad and\quad{\rm HMSE}_{2}\left(t\right):=\sum\limits _{\kappa\in\mathcal{K}}w_{\kappa}\,{\rm HMSE}_{2}\left(\boldsymbol{\lambda}_{\kappa},t\right).\label{HMSE}
\end{equation}
Specific formulas in \eqref{HMSE} can be found in Proposition \ref{prop.11}
(see the Appendix).

\subsection{The Case of Independent between frequency and individual severities}

First, we consider an independence between the frequency and individual
severities, i.e. $\beta_{0}^{[2]}=0$, while $b^{[1]}$ and $b^{[2]}$
is allowed to vary in nine scenarios that different correlations
are implied among the frequencies and the individual severities, respectively.
The comparison of MSEs is in Figure \ref{fig.numeric1}.

\begin{figure}[H]
\centering \includegraphics[width=1\textwidth]{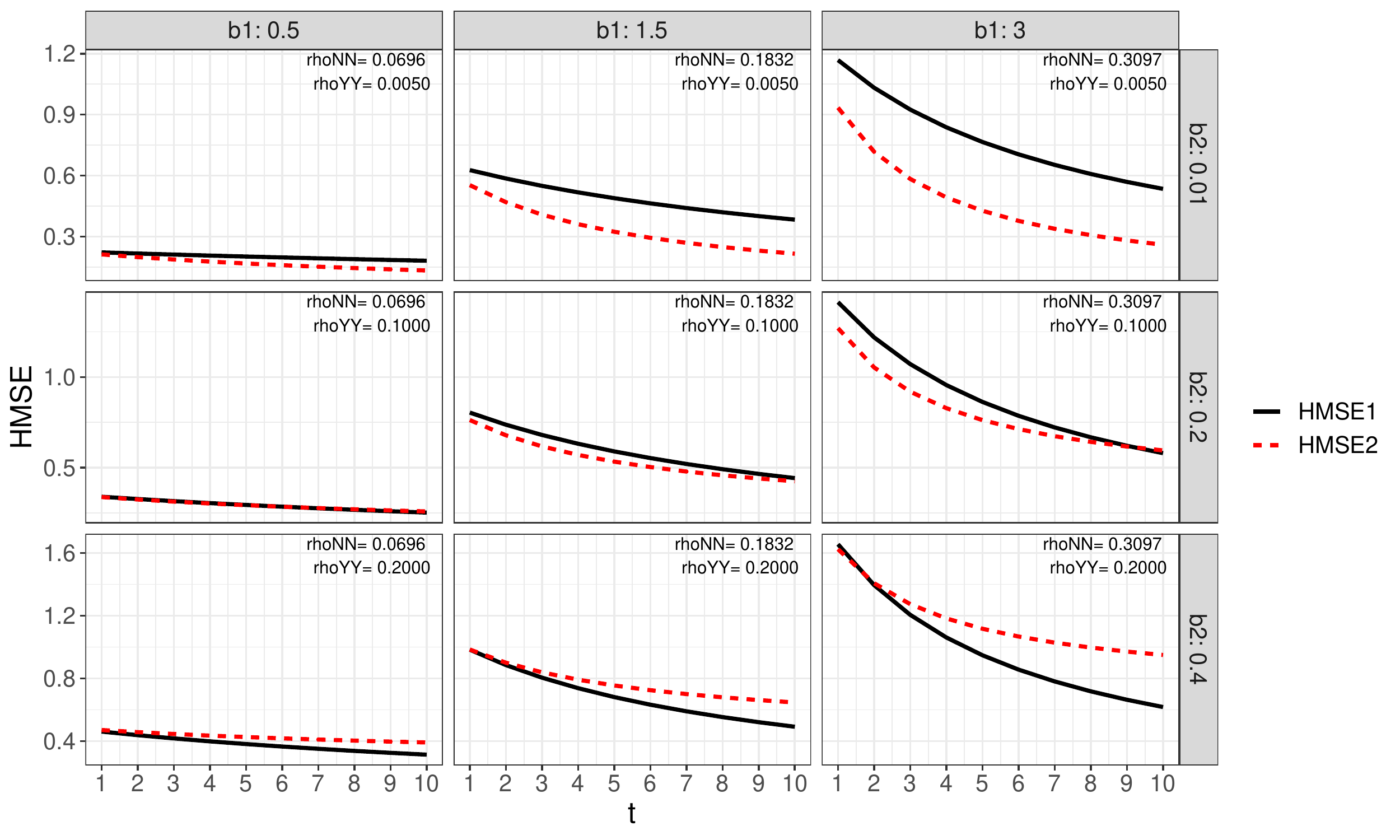}
\caption{The values of ${\rm HMSE}_{1}(t)$ and ${\rm HMSE}_{2}(t)$ for $\beta_{0}=0$}
\label{fig.numeric1}
\end{figure}

Theoretically, when there is \emph{a relatively weak dependence among
individual severities $b^{[2]}\approx0$ while a relatively strong
dependence among frequencies, $b^{[1]}>>0$}, we expect that the claim
history of frequency has the most predictive power for the premium,
while the additional information on the severities provides little
benefit. Especially, in the the extreme case of $b^{[2]}=0$ and $b^{[1]}>>0$,
$\Fn$ is the only valid information for predicting premiums, the
prediction $\E{\left[S_{t+1}|\Fa,\boldsymbol{\Lambda}\right]}$ suffers
from the loss of information. This is demonstrated by the following
comparison of premiums

\[
\E{\left[S_{t+1}|\mathcal{F}_{t}^{[agg]},\boldsymbol{\Lambda}\right]}=\Lambda^{[1]}\Lambda^{[2]}\E{\left[R^{^{[1]}}|\mathcal{F}_{t}^{[agg]},\boldsymbol{\Lambda}\right]}
\]
and
\[
\E{\left[S_{t+1}|\mathcal{F}_{t}^{[freq]},\boldsymbol{\Lambda}\right]}=\Lambda^{[1]}\Lambda^{[2]}\E{\left[R^{^{[1]}}|\mathcal{F}_{t}^{[freq]},\boldsymbol{\Lambda}\right]},
\]
where the equalities in both of expressions are from the assumption
$b^{[2]}=0$. Together with
\[
\E{\left[\left(R^{^{[1]}}-\E{\left[R^{^{[1]}}|\mathcal{F}_{t}^{[agg]},\boldsymbol{\Lambda}\right]}\right)^{2}|
\boldsymbol{\Lambda}\right]}\ge
\E{\left[\left(R^{^{[1]}}-\E{\left[R^{^{[1]}}|\mathcal{F}_{t}^{[freq]},\boldsymbol{\Lambda}\right]}\right)^{2}|\boldsymbol{\Lambda}\right]},
\]
it implies that
\[
{\rm HMSE}_{1}\left(\boldsymbol{\Lambda},t\right)\ge{\rm HMSE}_{2}\left(\boldsymbol{\Lambda},t\right).
\]
if B\"uhlmann premiums are not very different from  posteriori mean of the aggregate severity.
As shown in Figure \ref{fig.numeric1}, for the case where $(b^{[1]},b^{[2]})=(3,0.01)$
, ${\rm Prem}_{2}(\boldsymbol{\Lambda})$ outperforms ${\rm Prem}_{1}(\boldsymbol{\Lambda})$
consistently over time while the absolute values of HMSEs increase
as the variance of $R^{[1]}$ increases.

On the other hand, when there is \emph{a relatively strong dependence
among individual severities $b^{[2]}>0$ while relatively weak dependences
among frequencies $b^{[1]}=0$}, we expect, in theory, that only
the historical aggregate severity information provides meaningful information in the
prediction of premium. Especially in the extreme case $b^{[2]}>>00$
and $b^{[1]}=0$, $\mathcal{F}_{t}^{[freq]}$ does not provide any
information for the prediction of premium, while $\Fa$ can provide
some information for the prediction. Such a difference can be explained
by the fact that the equations
\[
\E{\left[S_{t+1}|\mathcal{F}_{t}^{[agg]},\boldsymbol{\Lambda}\right]}=\Lambda^{[1]}\Lambda^{[2]}\E{\left[R^{^{[2]}}|\mathcal{F}_{t}^{[agg]},\boldsymbol{\Lambda}\right]}
\]
and
\[
\begin{aligned}\E{\left[S_{t+1}|\mathcal{F}_{t}^{[freq]},\boldsymbol{\Lambda}\right]} & =\Lambda^{[1]}\Lambda^{[2]}\E{\left[R^{^{[2]}}|\mathcal{F}_{t}^{[freq]},\boldsymbol{\Lambda}\right]}\\
 & =\Lambda^{[1]}\Lambda^{[2]}
\end{aligned}
\]
together with the inequality
\[
\E{\left[\left(R^{^{[2]}}-\E{\left[R^{^{[2]}}|\mathcal{F}_{t}^{[agg]},\boldsymbol{\Lambda}\right]}\right)^{2}|\boldsymbol{\Lambda}\right]}\le\E{\left[\left(R^{^{[2]}}-\E{\left[R^{^{[2]}}|\mathcal{F}_{t}^{[freq]},\boldsymbol{\Lambda}\right]}\right)^{2}|\boldsymbol{\Lambda}\right]}
\]
implies that
\[
{\rm HMSE}_{1}\left(\boldsymbol{\Lambda},t\right)\le{\rm HMSE}_{2}\left(\boldsymbol{\Lambda},t\right).
\]
if B\"uhlmann premiums are not very different from  posteriori mean of the aggregate severity.
Indeed as shown in Figure \ref{fig.numeric1}, for the cases $(b^{[1]},b^{[2]})=(0.5,0.4)$,
${\rm Prem}_{1}(\boldsymbol{\Lambda})$ outperforms ${\rm Prem}_{2}(\boldsymbol{\Lambda})$
consistently over time.

Another interesting point is the asymptotic behaviour of ${\rm HMSE}_{1}\left(\boldsymbol{\Lambda},t\right)$
and ${\rm HMSE}_{2}\left(\boldsymbol{\Lambda},t\right)$ for larger
$t$ as seen in Figure \ref{fig.numeric1_1}. First, ${\rm HMSE}_{1}\left(\boldsymbol{\Lambda},t\right)$
converges to zero as $t$ increases. Such a convergence is an expected
result because of the convergence of $\frac{\sum\limits _{k=1}^{t}S_{k}}{t}\rightarrow\E{\left[S_{t}|\boldsymbol{R},\boldsymbol{\Lambda}\right]}$
in \eqref{eq:P_agg}. On the other hand, the convergence of ${\rm HMSE}_{2}\left(\boldsymbol{\Lambda},t\right)$
to zero is not guaranteed. This is also expected as the convergence
of
\[
\frac{\sum\limits _{k=1}^{t}\widetilde{S}_{k}\left(N_{k},\boldsymbol{\Lambda}\right)}{t}\rightarrow\E{\left[S_{t}|\boldsymbol{R},\boldsymbol{\Lambda}\right]}
\]
is not guaranteed in general. Instead,
\[
\frac{\sum\limits _{k=1}^{t}\widetilde{S}_{k}\left(N_{k},\boldsymbol{\Lambda}\right)}{t}
\]
converges to
\[
\begin{aligned}\E{\left[\widetilde{S}_{t}\left(N_{t},\boldsymbol{\Lambda}\right)\big\vert \boldsymbol{R},\boldsymbol{\Lambda}\right]} & =\Lambda^{[2]}\E{\left[N_{t}|\boldsymbol{R},\boldsymbol{\Lambda}\right]}\\
 & =\Lambda^{[1]}\Lambda^{[2]}R^{[1]}
\end{aligned}
\]
which further implies the convergence of the B\"uhlmann premium ${\rm Prem}_{2}(\boldsymbol{\Lambda})$
in \eqref{ahn.12} to
\[
\Lambda^{[1]}\Lambda^{[2]}R^{[1]}
\]
as the number of observations, $t$, increases. Hence, ${\rm HMSE}_{2}\left(\boldsymbol{\Lambda},t\right)$
in such case can be written as
\[
\lim\limits _{t\rightarrow\infty}{\rm HMSE}_{2}\left(\boldsymbol{\Lambda},t\right)=\E\bigg[\E\left[\left(\Lambda^{[1]}\Lambda^{[2]}R^{[1]}R^{[2]}-\Lambda^{[1]}\Lambda^{[2]}R^{[1]}\right)^{2}\right] \Big\vert\boldsymbol{\Lambda}\bigg]\ge 0,
\]
where the equality holds if $\P\left(R^{[2]}=1\right)=1$.

\begin{figure}[H]
\centering \includegraphics[width=1\textwidth]{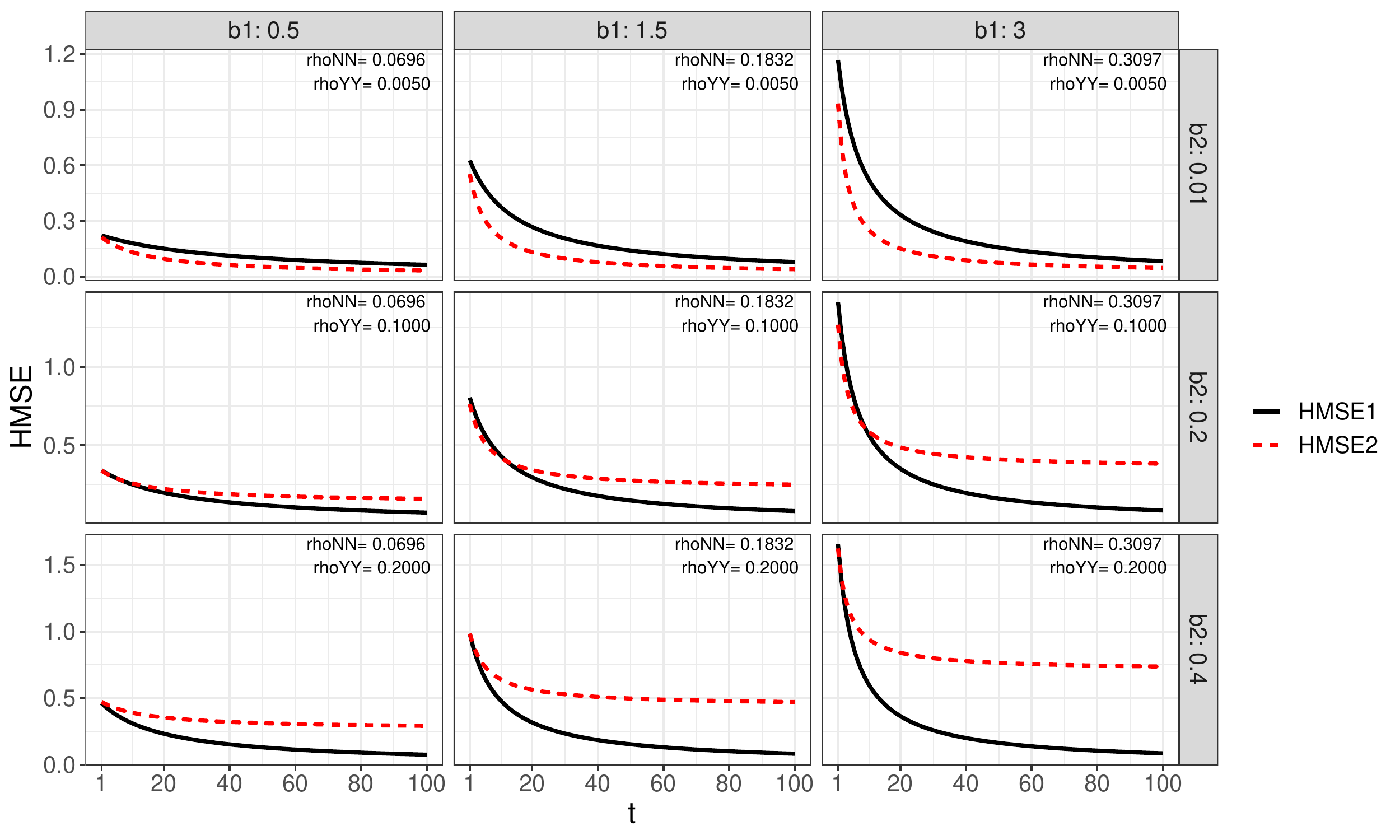}
\caption{The values of ${\rm HMSE}_{1}(t)$ and ${\rm HMSE}_{2}(t)$ for $\beta_{0}=0$ }
\label{fig.numeric1_1}
\end{figure}

\subsection{The Case of dependent between frequency and individual severities}

Here, we consider the case where both the frequency and individual
severities are dependent, which is the case $\beta_{0}^{[2]}\neq0$.
Motivated by the real data analysis in Section \ref{sec:Application}, a moderate dependence
with $\beta_{0}^{[2]}=-0.05$ and a relatively strong dependence with
$\beta_{0}^{[2]}=-0.1$ are assumed in Figures \ref{fig.numeric2}
and \ref{fig.numeric3} respectively. The overall patterns are similar
to the independent case $\beta_{0}^{[2]}=0$ that while\emph{ }a relatively
weak dependence among individual severities is combined with a relatively
strong dependence among frequencies\emph{, }the historical frequency information
has the most predictive power for the premium, and vice versa.

Specific results for HMSE are summarized in Table \ref{tab.numeric.hmse}.

\begin{table}[H]
\caption{(Numerical example) Mean square error $(10^{6})$}
\vspace{-0.05in}
 \centering \resizebox{\linewidth}{!}{
\begin{tabular}{lcccccccccccccccccccc}
\hline
 &  &  & $b^{[1]}$ & \multicolumn{3}{c}{0.5} &  & \multicolumn{3}{c}{1.5} &  & \multicolumn{3}{c}{3} &  &  &  &  &  & \tabularnewline
\hline
$\beta_{0}^{[2]}$ & $b^{[2]}$ &  & t & 1 & 5 & 10 &  & 1 & 5 & 10 &  & 1 & 5 & 10 &  &  &  &  &  & \tabularnewline
\hline
\hline
\multirow{6}{*}{0} & \multirow{2}{*}{0.01} & ${\rm HMSE}_{1}(t)$ &  & 0.1652 & 0.1509 & 0.1363 &  & 0.4325 & 0.3446 & 0.2748 &  & 0.7298 & 0.5031 & 0.3624 &  &  &  &  &  & \tabularnewline
 &  & ${\rm HMSE}_{2}(t)$ &  & 0.1579 & 0.1316 & 0.1180 &  & 0.3821 & 0.2650 & 0.2260 &  & 0.5878 & 0.3690 & 0.3171 &  &  &  &  &  & \tabularnewline
\cline{2-21} \cline{3-21} \cline{4-21} \cline{5-21} \cline{6-21} \cline{7-21} \cline{8-21} \cline{9-21} \cline{10-21} \cline{11-21} \cline{12-21} \cline{13-21} \cline{14-21} \cline{15-21} \cline{16-21} \cline{17-21} \cline{18-21} \cline{19-21} \cline{20-21} \cline{21-21}
 & \multirow{2}{*}{0.2} & ${\rm HMSE}_{1}(t)$ &  & 0.2565 & 0.2237 & 0.1928 &  & 0.5615 & 0.4214 & 0.3212 &  & 0.8941 & 0.5749 & 0.3975 &  &  &  &  &  & \tabularnewline
 &  & ${\rm HMSE}_{2}(t)$ &  & 0.2551 & 0.2288 & 0.2151 &  & 0.5326 & 0.4154 & 0.3764 &  & 0.8036 & 0.5848 & 0.5329 &  &  &  &  &  & \tabularnewline
\cline{2-21} \cline{3-21} \cline{4-21} \cline{5-21} \cline{6-21} \cline{7-21} \cline{8-21} \cline{9-21} \cline{10-21} \cline{11-21} \cline{12-21} \cline{13-21} \cline{14-21} \cline{15-21} \cline{16-21} \cline{17-21} \cline{18-21} \cline{19-21} \cline{20-21} \cline{21-21}
 & \multirow{2}{*}{0.4} & ${\rm HMSE}_{1}(t)$ &  & 0.3500 & 0.2915 & 0.2411 &  & 0.6916 & 0.4900 & 0.3592 &  & 1.0565 & 0.6364 & 0.4251 &  &  &  &  &  & \tabularnewline
 &  & ${\rm HMSE}_{2}(t)$ &  & 0.3573 & 0.3310 & 0.3173 &  & 0.6910 & 0.5738 & 0.5348 &  & 1.0307 & 0.8119 & 0.7601 &  &  &  &  &  & \tabularnewline
\hline
\multirow{6}{*}{-0.05} & \multirow{2}{*}{0.01} & ${\rm HMSE}_{1}(t)$ &  & 0.1913 & 0.1744 & 0.1570 &  & 0.5190 & 0.4092 & 0.3236 &  & 0.9161 & 0.6164 & 0.4376 &  &  &  &  &  & \tabularnewline
 &  & ${\rm HMSE}_{2}(t)$ &  & 0.1829 & 0.1482 & 0.1250 &  & 0.4580 & 0.2918 & 0.2224 &  & 0.7345 & 0.3967 & 0.2964 &  &  &  &  &  & \tabularnewline
\hline
 & \multirow{2}{*}{0.2} & ${\rm HMSE}_{1}(t)$ &  & 0.2951 & 0.2566 & 0.2206 &  & 0.6699 & 0.4973 & 0.3762 &  & 1.1156 & 0.6999 & 0.4775 &  &  &  &  &  & \tabularnewline
 &  & ${\rm HMSE}_{2}(t)$ &  & 0.2935 & 0.2587 & 0.2356 &  & 0.6353 & 0.4691 & 0.3997 &  & 1.0017 & 0.6639 & 0.5637 &  &  &  &  &  & \tabularnewline
\hline
 & \multirow{2}{*}{0.4} & ${\rm HMSE}_{1}(t)$ &  & 0.4014 & 0.3333 & 0.2750 &  & 0.8220 & 0.5761 & 0.4193 &  & 1.3125 & 0.7716 & 0.5092 &  &  &  &  &  & \tabularnewline
 &  & ${\rm HMSE}_{2}(t)$ &  & 0.4098 & 0.3751 & 0.3520 &  & 0.8219 & 0.6557 & 0.5863 &  & 1.2831 & 0.9453 & 0.8450 &  &  &  &  &  & \tabularnewline
\hline
\multirow{6}{*}{-0.1} & \multirow{2}{*}{0.01} & ${\rm HMSE}_{1}(t)$ &  & 0.2221 & 0.2019 & 0.1813 &  & 0.6270 & 0.4888 & 0.3833 &  & 1.1679 & 0.7651 & 0.5346 &  &  &  &  &  & \tabularnewline
 &  & ${\rm HMSE}_{2}(t)$ &  & 0.2125 & 0.1676 & 0.1332 &  & 0.5531 & 0.3238 & 0.2157 &  & 0.9339 & 0.4269 & 0.2596 &  &  &  &  &  & \tabularnewline
\hline
 & \multirow{2}{*}{0.2} & ${\rm HMSE}_{1}(t)$ &  & 0.3404 & 0.2951 & 0.2530 &  & 0.8046 & 0.5905 & 0.4431 &  & 1.4134 & 0.8633 & 0.5808 &  &  &  &  &  & \tabularnewline
 &  & ${\rm HMSE}_{2}(t)$ &  & 0.3385 & 0.2937 & 0.2593 &  & 0.7632 & 0.5340 & 0.4258 &  & 1.2701 & 0.7631 & 0.5958 &  &  &  &  &  & \tabularnewline
\hline
 & \multirow{2}{*}{0.4} & ${\rm HMSE}_{1}(t)$ &  & 0.4614 & 0.3819 & 0.3143 &  & 0.9835 & 0.6814 & 0.4924 &  & 1.6554 & 0.9480 & 0.6179 &  &  &  &  &  & \tabularnewline
 &  & ${\rm HMSE}_{2}(t)$ &  & 0.4713 & 0.4265 & 0.3920 &  & 0.9844 & 0.7552 & 0.6470 &  & 1.6240 & 1.1171 & 0.9497 &  &  &  &  &  & \tabularnewline
\hline
 &  &  &  &  &  &  &  &  &  &  &  &  &  &  &  &  &  &  &  & \tabularnewline
\end{tabular}} 
\label{tab.numeric.hmse}
\end{table}

\bigskip{}

\begin{figure}[hpt!]
\centering \includegraphics[width=1\textwidth]{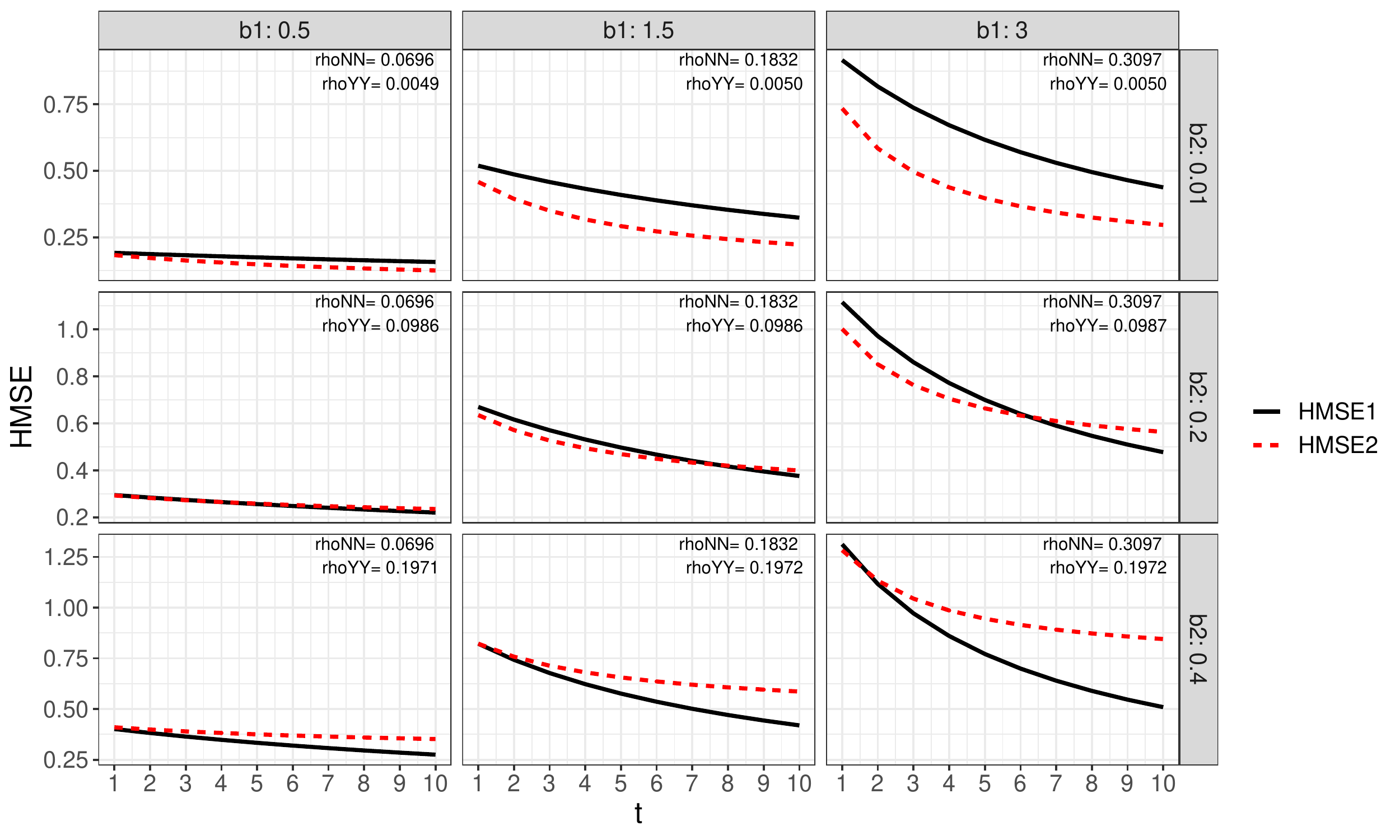}
\caption{The values of ${\rm HMSE}_{1}(t)$ and ${\rm HMSE}_{2}(t)$ for $\beta_{0}=-0.05$}
\label{fig.numeric2}
\end{figure}

\begin{figure}[H]
\centering \includegraphics[width=1\textwidth]{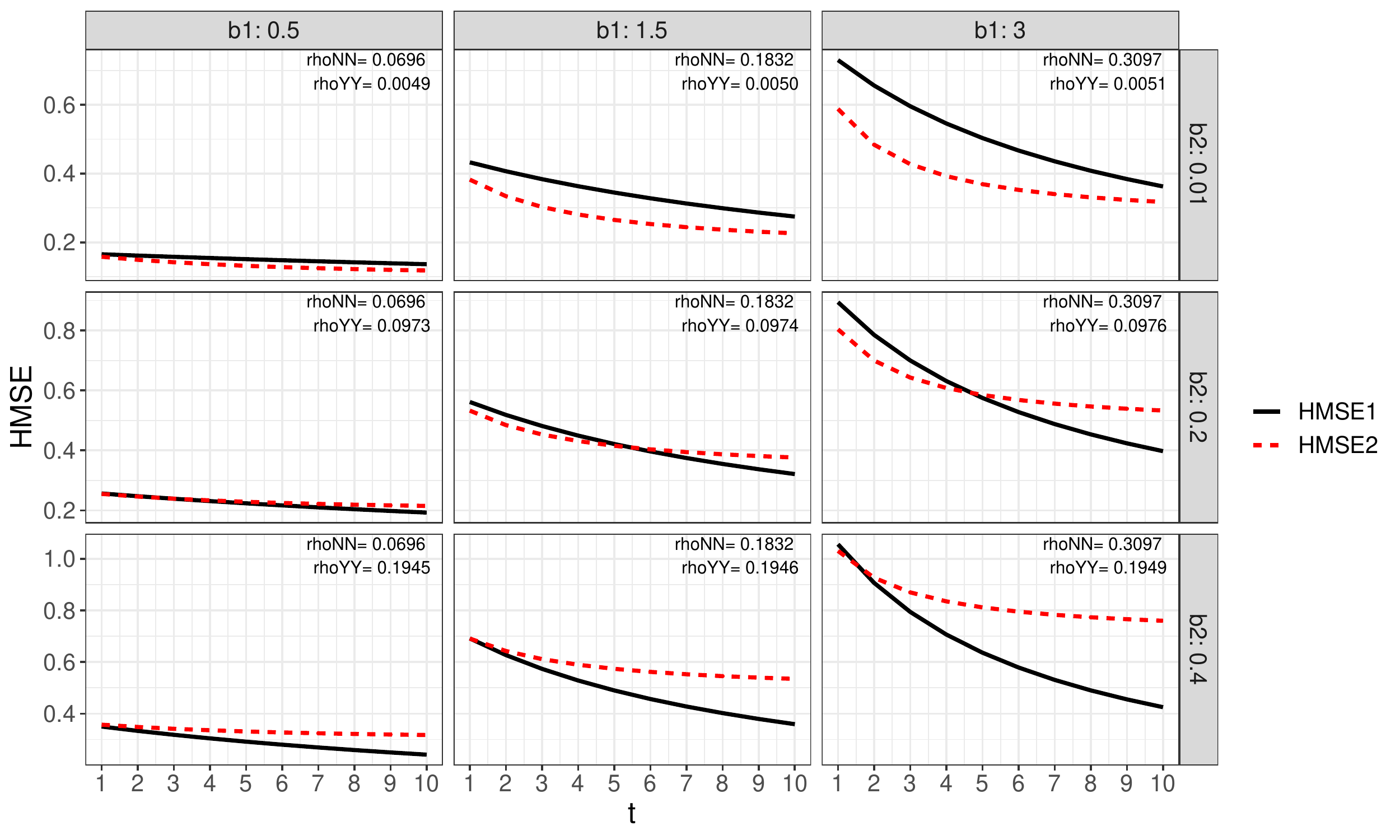}
\caption{The values of ${\rm HMSE}_{1}(t)$ and ${\rm HMSE}_{2}(t)$ for $\beta_{0}=-0.1$}
\label{fig.numeric3}
\end{figure}

In conclusion, we construct the following guidelines for practitioners
in choosing the appropriate posteriori risk classification approach
\begin{enumerate}
\item choose ${\rm Prem}_{1}(\boldsymbol{\Lambda})$ when there is a stronger
dependence among individual severities than that among frequencies
\item choose ${\rm Prem}_{2}(\boldsymbol{\Lambda})$ when there is a stronger
dependence among individual frequencies than that among individual
severities
\item the choice needs to be made dynamically over time as ${\rm HMSE}_{1}\left(\boldsymbol{\Lambda},t\right)\rightarrow0$
and ${\rm HMSE}_{2}\left(\boldsymbol{\Lambda},t\right)\rightarrow C>0.$
\end{enumerate}

\section{Application to Auto Insurance in Wisconsin Local Government Property
Insurance Fund\label{sec:Application}}

We illustrate our approach using data from the Wisconsin Local Government
Property Insurance Fund as in \citet{Frees4}. This fund offers insurance
protection for (i) property; (ii) motor vehicle; and (iii) contractors\textquoteright{}
equipment claims. Detailed information on the project is available
on the LGPIF project website. The LGPIF provides property insurance
for various governmental entities, including counties, cities, towns,
villages, school districts, fire departments, and other miscellaneous
entities. Collision coverage provides coverage for the impact of a
vehicle with an object, the impact of a vehicle with an attached
vehicle, or the overturn of a vehicle.

\subsection{Empirical Specification}

For the training sample data, we have used the longitudinal data from
1,234 local government entities cover from 2006 to 2010. We also
have hold-out sample data with 1098 observations from 379 local government
entities in the year of 2011. We removed the observations for policyholders
whose new collision coverage and old collision coverage are zero.
Hence, we use longitudinal data from 497 governmental entities in
our data analysis.\footnote{We adjust the values of the individual severity in the training sample
data so that the average individual severity in each year coincides
with the average individual severity in the year of 2011. 
} \textcolor{black}{We have two categorical variables:}
\begin{enumerate}
\item \textcolor{black}{the entity type with six levels, miscellaneous,
city, county, school, town and village, and average, }
\item \textcolor{black}{the coverage with three levels, coverage 1 $\in(0,0.14]$,
coverage 2 $\in(0.14,0.74]$, and coverage 3 $\in(0.74,\infty]$. }
\end{enumerate}
Under the settings in Model \ref{mod2}, we further assume
\[
R^{[1]}\sim{\rm IG}(1,b^{[1]}),\quad and\quad R^{[2]}\sim{\rm Gamma}(1,b^{[2]})
\]
so that
\[
\E\left[{R^{[1]}}\right]=\E{\left[R^{[2]}\right]}=1\quad and\quad\frac{\Var{\left[R^{[1]}\right]}}{b^{[1]}}=\frac{\Var\left[{R^{[2]}}\right]}{b^{[2]}}=1.
\]

\subsection{Estimation via Bayesian MCMC}

Our model specification in Section \ref{subsec:assump_CRM} is in
the form of multivariate nonlinear time-series with random effects,
that its estimation can be problematic in practice. Bayesian Econometric
methods (\citet{koop2003bayesian}) have made its popularity over
the last decade for its theoretical novelty and empirical performance,
especially for its application in economics and finance. The application
to the Actuarial research community has flourished over the last decades
(see \citet{klugman2013bayesian} and \citet{makov1996bayesian})
due to its intrinsic compatibility with Actuarial credibility theory.

To estimate the model under the Bayesian framework, we assume multivariate
Gaussian priors for the regression coefficients, i.e.

\[
\beta^{[1]}\sim {\rm MVN}(a_{0}^{[1]},A_{0}^{[1]})\quad\quad\beta^{[2]}\sim {\rm MVN}(a_{0}^{[2]},A_{0}^{[2]})
\]
and assume conjugate prior structure for
\[
\beta_{0}^{[2]}\sim {\rm N}(c_{0},d_{0}),\quad\psi^{[2]}\sim {\rm IGAM}(\alpha_{\psi},\delta_{\psi}),\quad b^{[1]}\sim {\rm IGAM}(\alpha_{b^{[2]}},\delta_{b^{[1]}})\; \hbox{and} \;b^{[2]}\sim {\rm IGAM}(\alpha_{b^{[2]}},\delta_{b^{[2]}}),
\]
where
$a_{0}^{[1]},A_{0}^{[1]},a_{0}^{[2]},A_{0}^{[2]},c_{0},d_{0},\alpha_{\psi},\delta_{\psi},\alpha_{b^{[2]}},\delta_{b^{[1]}},\alpha_{b^{[2]}},\delta_{b^{[2]}}$
are the prior hyper-parameters.
Note that  ${\rm IGAM}(\alpha, \sigma)$ is inverse gamma distribution with shape parameter $\alpha$ and scale parameter $\sigma$.

Due to the relatively complicated hierarchical structure, the posterior
distribution of the model parameters is not analytically feasible.
We reply to Markov Chain Monte-Carlo(MCMC) methods for obtaining empirical
estimates of the posterior statistics.\textcolor{red}{{} }The conjugate
prior specification gives known conditional likelihood that a simple
Gibbs sampler is used for estimating the parameters, $\beta_{0}^{[2]},\psi^{[2]},b^{[1]}$
and $b^{[2]}$, and a Metropolis-Hasting with random walk proposal
is used for estimating the coefficients. A more realistic prior structure
can be assumed, yet it implies a more computationally intensive MCMC
algorithm and possibly poor mixing. Hence, we alleviate this from
the current analysis.  For running MCMC, we use a
software, JAGS \citep{plummer2003jags}, that is a program for analysis of Bayesian hierarchical
models using MCMC. We have run 30,000 MCMC iterations
saving every 5th sample after burn-in of 20,000 iterations.
Multiple parallel MCMC chains are run to cross-validate the
convergence of the results.

Summary statistics of the posterior samples for the parameters in
Model \ref{mod2} using the Bayesian approach are presented in Table
\ref{est.model2} in \ref{apdx.tabs}. The table includes the posterior
median (EST), the posterior standard deviation (Std.dev), and the
95\% highest posterior density Bayesian credible interval (95$\%$
CI).  Note that a $*$ sign indicates the parameters whose 95 CI does not contain zero.

\begin{figure}[H]
\centering \includegraphics[width=0.6\textwidth]{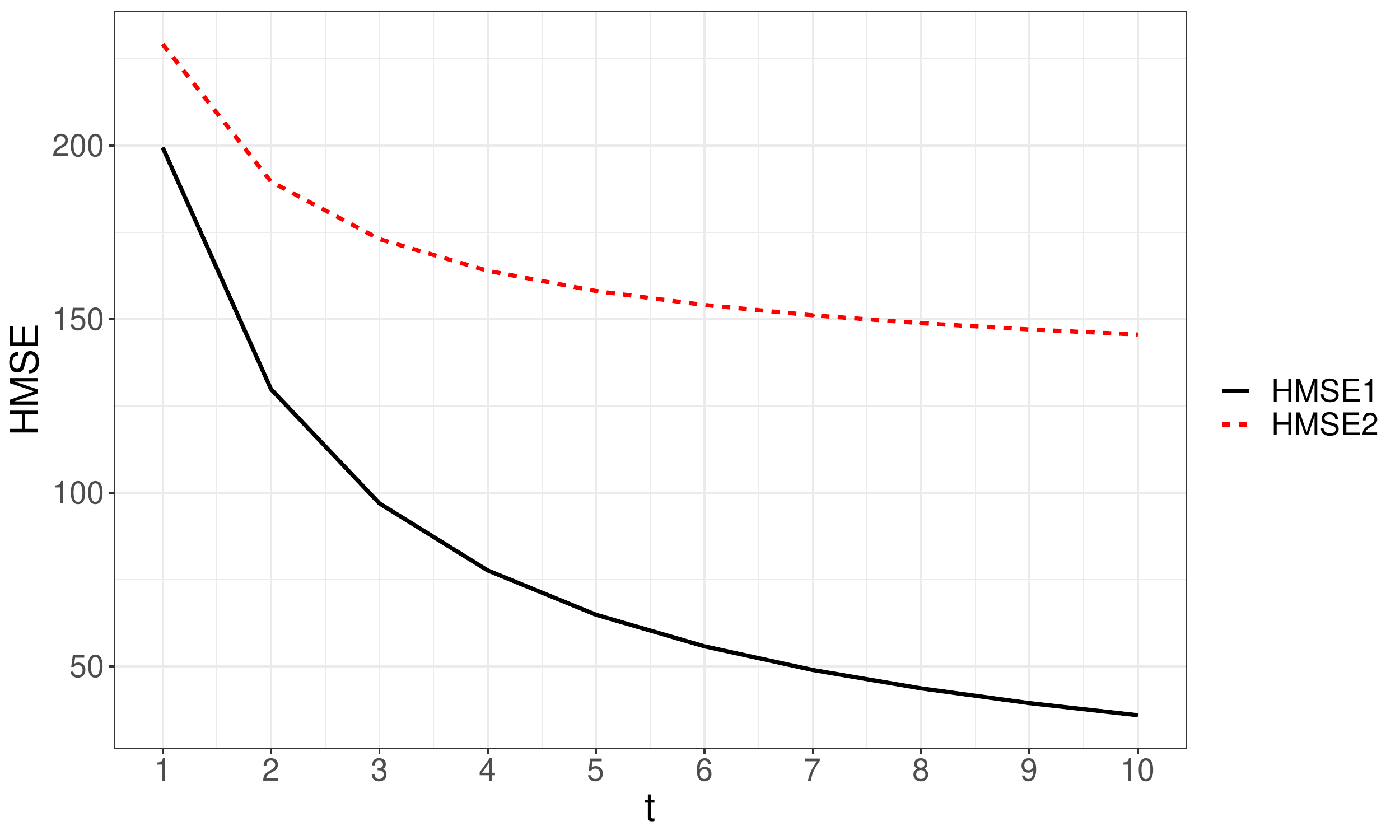}
\caption{The values of ${\rm HMSE}_{1}(t)$ and ${\rm HMSE}_{2}(t)$}
\label{fig.real.hmse}
\end{figure}

\begin{table}[H]
\caption{(Data example) Hypothetical Mean square error of Buhlmann premium
$(10^{6})$}
\vspace{-0.05in}
 \centering \resizebox{\linewidth}{!}{
\begin{tabular}{lccccccccccc}
\hline
t & 1 & 2 & 3 & 4 & 5 & 6 & 7 & 8 & 9 & 10 & \tabularnewline
\hline
${\rm HMSE}_{1}(t)$ & 199.46 & 129.83 & 96.92 & 77.60 & 64.83 & 55.74 & 48.92 & 43.62 & 39.37 & 35.88 & \tabularnewline
${\rm HMSE}_{2}(t)$ & 229.21 & 189.60 & 173.06 & 163.93 & 158.11 & 154.07 & 151.10 & 148.83 & 147.03 & 145.57 & \tabularnewline
\hline
\end{tabular}} 
\label{tab.real.hmse}
\end{table}

Figure \ref{fig.real.hmse} shows the comparison of the HMSE's of
two B\"uhlmann premiums, ${\rm Prem}_{1}(\boldsymbol{\Lambda})$ and
${\rm Prem}_{2}(\boldsymbol{\Lambda})$ in Section \ref{sec:Two-B=0000FChlmann-Premiums}.
The results are also summarized in Table \ref{tab.real.hmse}. In
terms of HMSE, B\"uhlmann premium ${\rm Prem}_{1}(\boldsymbol{\Lambda})$
outperforms B\"uhlmann premium ${\rm Prem}_{2}(\boldsymbol{\Lambda})$
regardless of the number of observations, while their gap becomes
larger as the number of observations increases. Moreover, while the
HMSE of B\"uhlmann premium ${\rm Prem}_{1}(\boldsymbol{\Lambda})$ asymptotically
converges to zero, the HMSE of B\"uhlmann premium ${\rm Prem}_{2}(\boldsymbol{\Lambda})$
converges to a non-zero constant. In conclusion, it is recommended
to use aggregate severity in the posteriori risk classification rather
than using the frequency. out-of-sample validation results in Table
\ref{tab.real.mse} show that ${\rm Prem}_{1}(\boldsymbol{\Lambda})$
outperforms ${\rm Prem}_{2}(\boldsymbol{\Lambda})$ consistently.

\begin{table}[H]
\caption{(Data example: validation) Mean square error of Buhlmann premium $(10^{6})$}
\vspace{-0.05in}
 \centering 
\begin{tabular}{lccccc}
\hline
t & 1 & 2 & 3 & 4 & 5\tabularnewline
\hline
${\rm MSE}_{1}(t)$ & 161.73 & 176.00 & 202.48 & 188.08 & 189.53\tabularnewline
${\rm MSE}_{2}(t)$ & 173.39 & 172.40 & 195.27 & 192.89 & 196.74\tabularnewline
\hline
\end{tabular}
\label{tab.real.mse}
\end{table}

\section{Remark on the Statistical Modelling of Collective Risk Model}

As briefly discussed in introduction, there are two ways of describing the CRM, i.e. the
\emph{two-part model} and the\emph{ direct model}, discussed in this
paper. While the latter demonstrate robust prediction of the mean
regardless of the parametric distribution used, however, the use of
the partial information on aggregate severity can be insufficient
in the estimation procedure and predictive analysis. On the other
hand, the \textit{two-part model} is sensitive to the model specification
meaning that its prediction ability is not guaranteed under the model
misspecification. Yet, when the model assumption is appropriate, it
shows better performance in the prediction of aggregate severity compared
to the \textit{direct model} since it uses the full information on
both the historical frequency and severities. As we have shown in
Section \ref{sec:Criteria}, using the historical aggregate severity information
can damage the prediction especially when the dependence among individual
severities is not so significant. Hence, our general suggestion for
the use of the direct approach is only for cases where there is a
relatively strong dependence among individual severities assuming
no non-statistical preferences.

\section*{Acknowledgements}

Jae Youn Ahn was supported by a National Research Foundation of Korea
(NRF) grant funded by the Korean Government (NRF-2017R1D1A1B03032318).
Rosy Oh was supported by Basic Science Research Program through the
National Research Foundation of Korea(NRF) funded by the Ministry
of Education (Grant No. 2019R1A6A1A11051177).

\section*{Appendix A: Important Lemmas }

\begin{lemma}\label{App:lem1} Consider the settings in Model \ref{mod1},
and let the conditional distribution of a random variable $M_{t}$
be given by
\begin{equation}
M_{t}\big\vert\left(\Z,N_{t}\right)\sim{\rm ED}\left(\Lambda^{[2]}\exp\left(\beta_{0}^{[2]}N_{t}\right)R^{[2]},\psi^{[2]}/N_{t}\right)\quad for\;t=1,2,...\label{eq:A1}
\end{equation}
for $N_{t}>0$, and
\[
\P\left(M_{t}=0\big\vert\Z,N_{t}\right)=1
\]
for $N_{t}=0$ based on the same EDF in \eqref{eq.y}. Then, we have
the distributional assumption on $Y_{t,j}$ in \eqref{eq.y} and the
distribution assumption on $M_{t}$ in \eqref{eq:A1} are equivalent
if the index set of the exponential dispersion family in \eqref{eq.y}
is $\Lambda=\R^{+}$.\end{lemma}

\begin{proof} First, note that the members of the model are infinitely divisible
if and only if the index set $\Lambda=\R^{+}$ \citep{jorgensen1997theory}.
Then, the proof follows from the reproductive property of EDF and
the infinitely divisible property of EDF.
\end{proof}

\begin{lemma}\label{dan.lem.1} Under the settings in Model \ref{mod2},
we have the moment generating function of $N_{t}$ is given by
\begin{equation}
M_{N_{t}}\left(z|\boldsymbol{R},\boldsymbol{\Lambda}\right)\equiv
\E\left[e^{zN_{t}}|\boldsymbol{R},\boldsymbol{\Lambda}\right]=\exp\left(\Lambda^{[1]}R^{[1]}\left(e^{z}-1\right)\right).\label{dan.10}
\end{equation}
Furthermore, we have
\[
M_{N_{t}}^{\prime}\left(z|\boldsymbol{R},\boldsymbol{\Lambda}\right)\equiv\E\left[{N_{t}e^{zN_{t}}|\boldsymbol{R},\boldsymbol{\Lambda}}\right]=\Lambda^{[1]}R^{[1]}e^{z}\exp\left(\Lambda^{[1]}R^{[1]}\left(e^{z}-1\right)\right)
\]
and
\[
\begin{aligned}M_{N_{t}}^{\prime\prime}\left(z|\boldsymbol{R},\boldsymbol{\Lambda}\right) & \equiv\E\left[{N_{t}^{2}e^{zN_{t}}|\boldsymbol{R},\boldsymbol{\Lambda}}\right]\\
 & =\left(\Lambda^{[1]}R^{[1]}\right)^{2}e^{2z}\exp\left(\Lambda^{[1]}R^{[1]}\left(e^{z}-1\right)\right)+\Lambda^{[1]}R^{[1]}e^{z}\exp\left(\Lambda^{[1]}R^{[1]}\left(e^{z}-1\right)\right).
\end{aligned}
\]
\end{lemma}

\begin{proof} The proof of \eqref{dan.10} is the moment generating function
of the Poisson distribution, and the other results follows by differentiating
\eqref{dan.10} with respect to $z$.
\end{proof}

\begin{lemma}\label{dan.lem.2} Under the settings in Model \ref{mod2},
if we assume that
\[
R^{[1]}\sim{\rm IG}(1,b^{[1]})
\]
we have the moment generating function of $R^{[1]}$ is given by
\begin{equation}
\begin{aligned}M_{R^{[1]}}\left(z\right) & :=\E\left[{e^{zR^{[1]}}}\right]\\
 & =\exp\left(\frac{1}{b^{[1]}}\left(1-\sqrt{1-2b^{[1]}z}\right)\right).
\end{aligned}
\label{dan.11}
\end{equation}
Furthermore, we also have
\[
\begin{aligned}M_{R^{[1]}}^{\prime}\left(z\right) & =\E\left[{R^{[1]}e^{zR^{[1]}}}\right]\\
 & =M_{R^{[1]}}\left(z\right)\left(1-2b^{[1]}z\right)^{-1/2}
\end{aligned}
\]
and
\[
\begin{aligned}M_{R^{[1]}}^{\prime\prime}\left(z\right) & =\E\left[\left(R^{[1]}\right)^{2}e^{zR^{[1]}}\right]\\
 & =M_{R^{[1]}}^{\prime}\left(z\right)\left[\left(1-2b^{[1]}z\right)^{-1/2}+b^{[1]}\left(1-2b^{[1]}z\right)^{-1}\right].%
\end{aligned}
\]
\end{lemma}

\begin{proof} The proof of \eqref{dan.11} is the moment generating function
of Inverse Gaussian distribution, and the other results follows by
differentiating \eqref{dan.11} with respect to $z$.

Here, we show how to derive B\"uhlmann premiums in Proposition \ref{prop:P_agg_2}
and \ref{Prop:P_freq_2}. First, we define
\[
u_{1}\left(\boldsymbol{\Lambda}\right):=\E{\left[u_{1}\left(\boldsymbol{R},\boldsymbol{\Lambda}\right)|\boldsymbol{\Lambda}\right]},\quad v_{1}\left(\boldsymbol{\Lambda}\right):=\E{\left[v_{1}\left(\boldsymbol{R},\boldsymbol{\Lambda}\right)|\boldsymbol{\Lambda}\right]},\quad \hbox{and} \quad a_{1}\left(\boldsymbol{\Lambda}\right):=\Var{\left[u_{1}\left(\boldsymbol{R},\boldsymbol{\Lambda}\right)|\boldsymbol{\Lambda}\right]},
\]
where
\[
u_{1}\left(\boldsymbol{R},\boldsymbol{\Lambda}\right):=\E\left[{S_{t}|\boldsymbol{R},\boldsymbol{\Lambda}}\right]\quad \hbox{and}\quad v_{1}\left(\boldsymbol{R},\boldsymbol{\Lambda}\right):=\Var{\left[S_{t}|\boldsymbol{R},\boldsymbol{\Lambda}\right]}.
\]
Similarly, define
\[
u_{2}\left(\boldsymbol{\Lambda}\right):=\E\left[{u_{2}\left(\boldsymbol{R},\boldsymbol{\Lambda}\right)|\boldsymbol{\Lambda}}\right],\quad v_{2}\left(\boldsymbol{\Lambda}\right):=\E{\left[v_{2}\left(\boldsymbol{R},\boldsymbol{\Lambda}\right)|\boldsymbol{\Lambda}\right]},\quad and\quad a_{2}\left(\boldsymbol{\Lambda}\right):=\Var{\left[u_{2}\left(\boldsymbol{R},\boldsymbol{\Lambda}\right)|\boldsymbol{\Lambda}\right]},
\]
where
\[
u_{2}\left(\boldsymbol{R},\boldsymbol{\Lambda}\right):=\E{\left[\E{\left[S_{t}|N_{t},\boldsymbol{\Lambda}\right]}|\boldsymbol{R},\boldsymbol{\Lambda}\right]}\quad \hbox{and} \quad v_{1}\left(\boldsymbol{R},\boldsymbol{\Lambda}\right):=\Var{\left[\E{\left[S_{t}|N_{t},\boldsymbol{\Lambda}\right]}|\boldsymbol{R},\boldsymbol{\Lambda}\right]}.
\]
\end{proof}

Analytical expression of
\[
u_{1}\left(\boldsymbol{\Lambda}\right),\quad v_{1}\left(\boldsymbol{\Lambda}\right),\quad a_{1}\left(\boldsymbol{\Lambda}\right),\quad u_{2}\left(\boldsymbol{\Lambda}\right),\quad v_{2}\left(\boldsymbol{\Lambda}\right),\quad and\quad a_{2}\left(\boldsymbol{\Lambda}\right)
\]
can be derived from Lemma \ref{dan.lem.st} and \ref{dan.lem.3} below.

\begin{lemma}\label{dan.lem.st} Under the settings in Model \ref{mod2},
we have
\[
u_{1}\left(\boldsymbol{R},\boldsymbol{\Lambda}\right)=\Lambda^{[1]}\Lambda^{[2]}R^{[1]}R^{[2]}e^{\beta_{0}^{[2]}}\exp\left(\Lambda^{[1]}R^{[1]}\left(e^{\beta_{0}^{[2]}}-1\right)\right)
\]
and
\[
\begin{aligned}v_{1}\left(\boldsymbol{R},\boldsymbol{\Lambda}\right) & =\left(1+\psi^{[2]}\right)\left(\Lambda^{[2]}R^{[2]}e^{\beta_{0}^{[2]}}\right)^{2}\Lambda^{[1]}R^{[1]}\exp\left(\Lambda^{[1]}R^{[1]}\left(e^{2\beta_{0}^{[2]}}-1\right)\right)\\
 & \quad\quad\quad\quad\quad\quad+\left(\Lambda^{[1]}R^{[1]}\Lambda^{[2]}R^{[2]}e^{2\beta_{0}^{[2]}}\right)^{2}\exp\left(\Lambda^{[1]}R^{[1]}\left(e^{2\beta_{0}^{[2]}}-1\right)\right)\\
%
 & \quad\quad\quad\quad\quad\quad\quad\quad-\left(\Lambda^{[1]}R^{[1]}\Lambda^{[2]}R^{[2]}e^{\beta_{0}^{[2]}}\right)^{2}\exp\left(2\Lambda^{[1]}R^{[1]}\left(e^{\beta_{0}^{[2]}}-1\right)\right).\\
%
%
\end{aligned}
\]
\end{lemma}
\begin{proof}  The first equation is the result
in Proposition \ref{Prop_P_agg_1}. For the proof of the second equation,
we have
\[
\begin{aligned}\\
v_{1}\left(\boldsymbol{R},\boldsymbol{\Lambda}\right) & =\Var{\left[S_{t}|\boldsymbol{R},\boldsymbol{\Lambda}\right]}\\
 & =\E{\left[\Var{\left[N_{t}M_{t}|N_{t},\boldsymbol{R},\boldsymbol{\Lambda}\right]}\Big\vert\boldsymbol{R},\boldsymbol{\Lambda}\right]}+\Var{\left[\E{\left[N_{t}M_{t}|N_{t},\boldsymbol{R},\boldsymbol{\Lambda}\right]}\Big\vert\boldsymbol{R},\boldsymbol{\Lambda}\right]}\\
 & =\E{\left[N_{t}\left(\Lambda^{[2]}R^{[2]}e^{\beta_{0}^{[2]}N_{t}}\right)^{2}\psi^{[2]}\Big\vert\boldsymbol{R},\boldsymbol{\Lambda}\right]}+\Var{\left[N_{t}\Lambda^{[2]}R^{[2]}e^{\beta_{0}^{[2]}N_{t}}\Big\vert\boldsymbol{R},\boldsymbol{\Lambda}\right]}\\
 & =\E{\left[\psi^{[2]}\left(\Lambda^{[2]}R^{[2]}\right)^{2}N_{t}\exp\left(2\beta_{0}^{[2]}N_{t}\right)\Big\vert\boldsymbol{R},\boldsymbol{\Lambda}\right]}+\Var{\left[\Lambda^{[2]}R^{[2]}N_{t}\exp\left(\beta_{0}^{[2]}N_{t}\right)\Big\vert\boldsymbol{R},\boldsymbol{\Lambda}\right]}\\
 & =\psi^{[2]}\left(\Lambda^{[2]}R^{[2]}\right)^{2}M_{N_{t}}^{\prime}\left(2\beta_{0}^{[2]}|\boldsymbol{R},\boldsymbol{\Lambda}\right)+\left(\Lambda^{[2]}R^{[2]}\right)^{2}\left[M_{N_{t}}^{\prime\prime}\left(2\beta_{0}^{[2]}|\boldsymbol{R},\boldsymbol{\Lambda}\right)-\left\{ M_{N_{t}}^{\prime}\left(\beta_{0}^{[2]}|\boldsymbol{R},\boldsymbol{\Lambda}\right)\right\} ^{2}\right]
\\
\end{aligned}
\]
which finishes the proof with Lemma \ref{dan.lem.1}.
\end{proof}

\begin{lemma}\label{dan.lem.3} Under the settings in Model \ref{mod2},
we have
\[
u_{2}\left(\boldsymbol{R},\boldsymbol{\Lambda}\right)=e^{\beta_{0}^{[2]}}\Lambda^{[1]}\Lambda^{[2]}R^{[1]}\exp\left(\Lambda^{[1]}R^{[1]}\left(e^{\beta_{0}^{[2]}}-1\right)\right)
\]
and
\[
\begin{aligned}v_{2}\left(\boldsymbol{R},\boldsymbol{\Lambda}\right)%
%
 & =\left(\Lambda^{[2]}\right)^{2}\left(\Lambda^{[1]}R^{[1]}\right)^{2}e^{4\beta_{0}^{[2]}}\exp\left(\Lambda^{[1]}R^{[1]}\left(e^{2\beta_{0}^{[2]}}-1\right)\right)\\
 & \quad\quad\quad\quad+\left(\Lambda^{[2]}\right)^{2}\Lambda^{[1]}R^{[1]}e^{2\beta_{0}^{[2]}}\exp\left(\Lambda^{[1]}R^{[1]}\left(e^{2\beta_{0}^{[2]}}-1\right)\right)\\
 & \quad\quad\quad\quad-\left(\Lambda^{[2]}\right)^{2}\left(\Lambda^{[1]}R^{[1]}\right)^{2}e^{2\beta_{0}^{[2]}}\exp\left(2\Lambda^{[1]}R^{[1]}\left(e^{\beta_{0}^{[2]}}-1\right)\right).
\end{aligned}
\]

\end{lemma}
\begin{proof} The first equation is from
\[
\begin{aligned}u_{2}\left(\boldsymbol{R},\boldsymbol{\Lambda}\right) & =\E\left[\Lambda^{[2]}N_{t}\exp\left(\beta_{0}^{[2]}N_{t}\right)|\boldsymbol{R},\boldsymbol{\Lambda}\right]\\
 & =\Lambda^{[2]}\E\left[N_{t}\exp\left(\beta_{0}^{[2]}N_{t}\right)|\boldsymbol{R},\boldsymbol{\Lambda}\right]\\
 & =e^{\beta_{0}^{[2]}}\Lambda^{[1]}\Lambda^{[2]}R^{[1]}\exp\left(\Lambda^{[1]}R^{[1]}\left(e^{\beta_{0}^{[2]}}-1\right)\right)
\end{aligned}
\]
where the second equality is from Lemma \ref{dan.lem.1} in Appendix.
For the proof of the second equation, we have
\[
\begin{aligned}v_{2}\left(\boldsymbol{R},\boldsymbol{\Lambda}\right) & =\Var\left[\Lambda^{[2]}N_{t}\exp\left(\beta_{0}^{[2]}N_{t}\right)|\boldsymbol{R},\boldsymbol{\Lambda}\right]\\
 & =\left(\Lambda^{[2]}\right)^{2}\Bigg[
 \E\left[N_{t}^{2}\exp\left(2\beta_{0}^{[2]}N_{t}\right)|\boldsymbol{R},\boldsymbol{\Lambda}\right]-
 \left(E\left[N_{t}\exp\left(\beta_{0}^{[2]}N_{t}\right)|\boldsymbol{R},\boldsymbol{\Lambda}\right]\right)^{2}\Bigg]\\
 & =\left(\Lambda^{[2]}\right)^{2}\left(\Lambda^{[1]}R^{[1]}\right)^{2}e^{4\beta_{0}^{[2]}}\exp\left(\Lambda^{[1]}R^{[1]}\left(e^{2\beta_{0}^{[2]}}-1\right)\right)\\
 & \quad\quad\quad\quad+\left(\Lambda^{[2]}\right)^{2}\Lambda^{[1]}R^{[1]}e^{2\beta_{0}^{[2]}}\exp\left(\Lambda^{[1]}R^{[1]}\left(e^{2\beta_{0}^{[2]}}-1\right)\right)\\
 & \quad\quad\quad\quad-\left(\Lambda^{[2]}\right)^{2}\left(\Lambda^{[1]}R^{[1]}\right)^{2}e^{2\beta_{0}^{[2]}}\exp\left(2\Lambda^{[1]}R^{[1]}\left(e^{\beta_{0}^{[2]}}-1\right)\right),
\end{aligned}
\]
where the third equality is from Lemma \ref{dan.lem.1} in Appendix.
\end{proof}

Then, following the procedure in B\"uhlmann premium, under Model \ref{mod2},
we have the conditional mean $\E{\left[S_{t}|\boldsymbol{\Lambda}\right]}$ and
the B\"uhlmann factor, $Z_{1}(\Lambda)$, can be expressed as
\[
\E\left[S_{t}|\boldsymbol{\Lambda}\right]=u_{1}\left(\boldsymbol{\Lambda}\right)\quad and\quad Z_{1}(\boldsymbol{\Lambda})=\frac{ta_{1}\left(\boldsymbol{\Lambda}\right)}{ta_{1}\left(\boldsymbol{\Lambda}\right)+v_{1}\left(\boldsymbol{\Lambda}\right)}.
\]
Similarly, the conditional mean $\E{\left[\E{\left[S_{t}|N_{t},\boldsymbol{\Lambda}\right]}|\boldsymbol{R},\boldsymbol{\Lambda}\right]}$
and B\"uhlmann factor can be expressed as
\[
\E\left[\E\left[S_{t}|N_{t},\boldsymbol{\Lambda}\right]|\boldsymbol{R},\boldsymbol{\Lambda}\right]=u_{2}\left(\boldsymbol{\Lambda}\right)\quad and\quad Z_{2}(\boldsymbol{\Lambda})=\frac{ta_{2}\left(\boldsymbol{\Lambda}\right)}{ta_{2}\left(\boldsymbol{\Lambda}\right)+v_{2}\left(\boldsymbol{\Lambda}\right)}.
\]

\section*{Appendix B: Auxiliary Results for the numerical illustration}

In the following proposition, we provide the analytical expressions
of useful statistics in Model \ref{mod2}. Note that the conditional
expressions are of primary interest to insurers because a priori information
is usually available at the time of the contract.

First, we provide the auxiliary lemma which is necessary for Proposition
\ref{prop.1} and the calculation of MSE.

\begin{lemma}\label{dan.lem.30} Consider the settings in Model \ref{mod2}.
Then, we have
\[
\E{\left[N_{t}^{2}\exp\left(zN_{t}\right)|\boldsymbol{\Lambda}\right]}=\left(\Lambda^{[1]}\right)^{2}e^{2z}M_{R^{[1]}}^{\prime\prime}\left(\Lambda^{[1]}\left(e^{z}-1\right)\right)+\Lambda^{[1]}e^{z}M_{R^{[1]}}^{\prime}\left(\Lambda^{[1]}\left(e^{z}-1\right)\right)
\]
and
\[
\E\left[N_{t}\exp\left(zN_{t}\right)R^{[1]}\exp\left(R^{[1]}\Lambda^{[1]}\left(e^{z}-1\right)\right)
|\boldsymbol{\Lambda}\right]=\Lambda^{[1]}e^{z}M_{R^{[1]}}^{\prime\prime}\left(2\Lambda^{[1]}\left(e^{z}-1\right)\right).
\]

For $t_{1}\neq t_{2}$, we have
\[
\E\left[{N_{t_{1}}N_{t_{2}}\exp\left(z\left(N_{t_{1}}+N_{t_{2}}\right)\right)|\boldsymbol{\Lambda}}\right]=\left(\Lambda^{[1]}\right)^{2}e^{2z}M_{R^{[1]}}^{\prime\prime}\left(2\Lambda^{[1]}\left(e^{z}-1\right)\right).
\]
Finally, we have
\[
\E\left[e^{zN_{t}}|\boldsymbol{\Lambda}\right]=M_{R^{[1]}}\left(\Lambda^{[1]}\left(e^{z}-1\right)\right)
\]
and
\[
\E\left[N_{t}e^{zN_{t}}|\boldsymbol{\Lambda}\right]=\Lambda^{[1]}e^{z}M_{R^{[1]}}^{\prime}\left(\Lambda^{[1]}\left(e^{z}-1\right)\right).
\]

\end{lemma}

\begin{proof} For the proof of the first equation, we have
\[
\begin{aligned}\E{\left[N_{t}^{2}\exp\left(zN_{t}\right)|\boldsymbol{\Lambda}\right]} & =\E{\left[\E{\left[N_{t}^{2}\exp\left(zN_{t}\right)|\boldsymbol{\Lambda},\boldsymbol{R}\right]}|\boldsymbol{\Lambda}\right]}\\
 & =\E{\left[\left(\Lambda^{[1]}R^{[1]}\right)^{2}e^{2z}\exp\left(\Lambda^{[1]}R^{[1]}\left(e^{z}-1\right)\right)+\Lambda^{[1]}R^{[1]}e^{z}\exp\left(\Lambda^{[1]}R^{[1]}\left(e^{z}-1\right)\right)|\boldsymbol{\Lambda}\right]}\\
 & =\left(\Lambda^{[1]}\right)^{2}e^{2z}M_{R^{[1]}}^{\prime\prime}\left(\Lambda^{[1]}\left(e^{z}-1\right)\right)+\Lambda^{[1]}e^{z}M_{R^{[1]}}^{\prime}\left(\Lambda^{[1]}\left(e^{z}-1\right)\right),
\end{aligned}
\]
where the second equality is from Lemma \ref{dan.lem.1}. For the
second equation, we have
\[
\begin{aligned}\E\left[{N_{t}\exp\left(zN_{t}\right)R^{[1]}\exp\left(R^{[1]}\Lambda^{[1]}\left(e^{z}-1\right)\right)|\boldsymbol{\Lambda}}\right]%
 & =\E\left[{R^{[1]}\exp\left(R^{[1]}\Lambda^{[1]}\left(e^{z}-1\right)\right)\E{\left[N_{t}\exp\left(zN_{t}\right)|\boldsymbol{\Lambda},\boldsymbol{R}\right]}|\boldsymbol{\Lambda}}\right]\\
 & =\E{\left[\Lambda^{[1]}\left(R^{[1]}\right)^{2}e^{z}\exp\left(2R^{[1]}\Lambda^{[1]}\left(e^{z}-1\right)\right)|\boldsymbol{\Lambda}\right]}\\
 & =\Lambda^{[1]}e^{z}M_{R^{[1]}}^{\prime\prime}\left(2\Lambda^{[1]}\left(e^{z}-1\right)\right),
\end{aligned}
\]
where the second equality is from Lemma \ref{dan.lem.1}. Finally,
for $t_{1}\neq t_{2}$, we have
\[
\begin{aligned}\E{\left[N_{t_{1}}N_{t_{2}}\exp\left(z\left(N_{t_{1}}+N_{t_{2}}\right)\right)|\boldsymbol{\Lambda}\right]} & =\E{\left[\E{\left[N_{t_{1}}N_{t_{2}}\exp\left(z\left(N_{t_{1}}+N_{t_{2}}\right)\right)|\boldsymbol{\Lambda},\boldsymbol{R}\right]}|\boldsymbol{\Lambda}\right]}\\
 & =\E{\left[\E{\left[N_{t_{1}}\exp\left(zN_{t_{1}}\right)|\boldsymbol{\Lambda},\boldsymbol{R}\right]}
 \E{\left[N_{t_{2}}\exp\left(zN_{t_{2}}\right)|\boldsymbol{\Lambda},\boldsymbol{R}\right]}|\boldsymbol{\Lambda}\right]}\\
 & =\E{\left[\left(\Lambda^{[1]}R^{[1]}\right)^{2}e^{2z}\exp\left(2R^{[1]}\Lambda^{[1]}\left(e^{z}-1\right)\right)|\boldsymbol{\Lambda}\right]}\\
 & =\left(\Lambda^{[1]}\right)^{2}e^{2z}M_{R^{[1]}}^{\prime\prime}\left(2\Lambda^{[1]}\left(e^{z}-1\right)\right).
\end{aligned}
\]
where the second equation is the conditional independence between
$N_{t_{1}}$ and $N_{t_{2}}$, and the third equality is from Lemma
\ref{dan.lem.1}.

Finally, for the proof of the last part, we have
\[
\begin{aligned}\E{\left[e^{zN_{t}}|\boldsymbol{\Lambda}\right]} & =\E{\left[\exp\left(\Lambda^{[1]}R^{[1]}\left(e^{z}-1\right)\right)|\boldsymbol{\Lambda}\right]}\\
 & =M_{R^{[1]}}\left(\Lambda^{[1]}\left(e^{z}-1\right)\right)
\end{aligned}
\]
and
\[
\begin{aligned}\E\left[{N_{t}e^{zN_{t}}|\boldsymbol{\Lambda}}\right] & =\E\left[{\E{\left[N_{t}e^{zN_{t}}|\boldsymbol{R},\boldsymbol{\Lambda}\right]}|\boldsymbol{\Lambda}}\right]\\
 & =\Lambda^{[1]}e^{z}\E{\left[R^{[1]}\exp\left(\Lambda^{[1]}R^{[1]}\left(e^{z}-1\right)\right)|\boldsymbol{\Lambda}\right]}\\
 & =\Lambda^{[1]}e^{z}M_{R^{[1]}}^{\prime}\left(\Lambda^{[1]}\left(e^{z}-1\right)\right),
\end{aligned}
\]
where the second equality is from Lemma \ref{dan.lem.1}.
\end{proof}

\begin{proposition}\label{prop.1}

Under Model \ref{mod2}, we have the following conditional expressions.
\begin{enumerate}
\item The mean and variance of the aggregate severity are
\[
\E\left[{S_{t}\Big\vert\boldsymbol{\Lambda}}\right]=\Lambda^{[1]}\Lambda^{[2]}e^{\beta_{0}^{[2]}}
\,M_{R^{[1]}}^{\prime}\left(\zeta_{1}\right)
\]
and
\begin{equation}
\begin{aligned}\Var{\left[S_{t}\Big\vert\boldsymbol{\Lambda}\right]} & =\Lambda^{[1]}\left(\Lambda^{[2]}\right)^{2}\left(1+b^{[2]}\right)e^{2\beta_{0}^{[2]}}\Bigg[\left(1+\psi^{[2]}\right)M_{R^{[1]}}^{\prime}\left(\zeta_{2}\right)+\Lambda^{[1]}e^{2\beta_{0}^{[2]}}M_{R^{[1]}}^{\prime\prime}\left(\zeta_{2}\right)\Bigg]\\
 & \quad-\bigg\{\Lambda^{[1]}\Lambda^{[2]}e^{\beta_{0}^{[2]}}M_{R^{[1]}}^{\prime}\left(\zeta_{1}\right)\bigg\}^{2}.
\end{aligned}
\label{eq.201}
\end{equation}
\item The covariance of aggregate severities is
\[
\cov\left[{S_{t_{1}},S_{t_{2}}\Big\vert\boldsymbol{\Lambda}}\right]=\left(\Lambda^{[1]}\Lambda^{[2]}\right)^{2}e^{2\beta_{0}^{[2]}}\Bigg[\left(1+b^{[2]}\right)M_{R^{[1]}}^{\prime\prime}\left(2\zeta_{1}\right)-\left\{ M_{R^{[1]}}^{\prime}\left(\zeta_{1}\right)\right\} ^{2}\Bigg]
\]
for $t_{1}\neq t_{2}$.
\item The covariance among the frequencies is
\[
\cov\left[{N_{t_{1}},N_{t_{2}}\Big\vert\boldsymbol{\Lambda}}\right]=\left(\Lambda^{[1]}\right)^{2}b^{[1]}
\]
for $t_{1}\neq t_{2}$.
\item The variance of the individual severities is
\[
\Var\left[{Y_{t,j}\Big\vert\boldsymbol{\Lambda}}\right]=\left(\Lambda^{[2]}\right)^{2}\bigg[\left(1+b^{[2]}\right)\left(1+\psi^{[2]}\right)\,M_{R^{[1]}}\left(\zeta_{2}\right)-\left\{ M_{R^{[1]}}\left(\zeta_{1}\right)\right\} ^{2}\bigg].
\]
\item The covariances among the individual severities are

\[
\cov\left[{Y_{t,j_{1}},Y_{t,j_{2}}\Big\vert\boldsymbol{\Lambda}}\right]=\left(\Lambda^{[2]}\right)^{2}\bigg[\left(1+b^{[2]}\right)M_{R^{[1]}}\left(\zeta_{2}\right)-\left\{ M_{R^{[1]}}\left(\zeta_{1}\right)\right\} ^{2}\bigg]
\]
and, for $t_{1}\neq t_{2}$,
\[
\cov\left[{Y_{t_{1},j_{1}},Y_{t_{2},j_{2}}\bigg\vert\boldsymbol{\Lambda}}\right]=\left(\Lambda^{[2]}\right)^{2}\bigg[\left(1+b^{[2]}\right)M_{R^{[1]}}\left(2\zeta_{1}\right)-\left\{ M_{R^{[1]}}\left(\zeta_{1}\right)\right\} ^{2}\bigg].
\]

\item The covariances among the individual severities are
\[
\cov\left[{N_{t},Y_{t,j}\bigg\vert\boldsymbol{\Lambda}}\right]=\Lambda^{[1]}\Lambda^{[2]}\bigg[e^{\beta_{0}^{[2]}}M_{R^{[1]}}^{\prime}\left(\zeta_{1}\right)-M_{R^{[1]}}\left(\zeta_{1}\right)\bigg]
\]
and, for $t_{1}\neq t_{2}$,
\[
\cov\left[{N_{t_{1}},Y_{t_{2},j_{1}}\bigg\vert\boldsymbol{\Lambda}}\right]=0.
\]
\end{enumerate}
\end{proposition}

\begin{proof} Assume Model 2 and use conditional expectation, covariance,
and variance. Then, the mean and variance of the aggregate severity
conditional on the priori premium are calculated as

\[
\begin{aligned}\E\left[{S_{t}\Big\vert\boldsymbol{\Lambda}}\right] & =\E\left[{\E{\left[S_{t}\big\vert N_{t},\boldsymbol{\Lambda},\boldsymbol{R}\right]}\big\vert\boldsymbol{\Lambda}}\right]\\
 & =\Lambda^{[1]}\Lambda^{[2]}e^{\beta_{0}^{[2]}}\E{\left[R^{[1]}\exp\left(\Lambda^{[1]}R^{[1]}\left(e^{\beta_{0}^{[2]}}-1\right)\right)\big\vert\boldsymbol{\Lambda}\right]}\\
 & =\Lambda^{[1]}\Lambda^{[2]}e^{\beta_{0}^{[2]}}
\,M_{R^{[1]}}^{\prime}\left(\zeta_{1}\right),
\end{aligned}
\]
where the second equality comes from Proposition \ref{Prop_P_agg_1},
and
\[
\begin{aligned}\\
\Var\left[{S_{t}\Big\vert\boldsymbol{\Lambda}}\right] & =\E\left[{\Var{\left[S_{t}\big\vert N_{t},\boldsymbol{\Lambda},\boldsymbol{R}\right]}\Big\vert\boldsymbol{\Lambda}}\right]+\Var\left[{\E{\left[S_{t}\big\vert N_{t},\boldsymbol{\Lambda},\boldsymbol{R}\right]}\Big\vert\boldsymbol{\Lambda}}\right]\\
 & =\E\left[{N_{t}\left(\Lambda^{[2]}R^{[2]}\right)^{2}\psi^{[2]}\exp\left(2\beta_{0}^{[2]}N_{t}\right)\Big\vert\boldsymbol{\Lambda}}\right]+\Var\left[{N_{t}\Lambda^{[2]}R^{[2]}\exp\left(\beta_{0}^{[2]}N_{t}\right)\Big\vert\boldsymbol{\Lambda}}\right]\\
 & =\left(\Lambda^{[2]}\right)^{2}\E{\left[\left(R^{[2]}\right)^{2}\right]}\psi^{[2]}\E\left[{N_{t}\exp\left(2\beta_{0}^{[2]}N_{t}\right)\Big\vert\boldsymbol{\Lambda}}\right]\\
 & \quad+\left(\Lambda^{[2]}\right)^{2}\left[\E{\left[\left(R^{[2]}\right)^{2}\right]}\E{\left[N_{t}^{2}\exp\left(2\beta_{0}^{[2]}N_{t}\right)\Big\vert\boldsymbol{\Lambda}\right]}-\left\{ \E{\left[N_{t}\exp\left(\beta_{0}^{[2]}N_{t}\right)\Big\vert\boldsymbol{\Lambda}\right]}\right\} ^{2}\right]\\
 & =\left(\Lambda^{[2]}\right)^{2}\left(1+b^{[2]}\right)\psi^{[2]}\Lambda^{[1]}e^{2\beta_{0}^{[2]}}M_{R^{[1]}}^{\prime}\left(\Lambda^{[1]}\left(e^{2\beta_{0}^{[2]}}-1\right)\right)\\
 & \quad+\left(\Lambda^{[2]}\right)^{2}\Bigg[\left(1+b^{[2]}\right)\Big\{\left(\Lambda^{[1]}\right)^{2}e^{4\beta_{0}^{[2]}}M_{R^{[1]}}^{\prime\prime}\left(\Lambda^{[1]}\left(e^{2\beta_{0}^{[2]}}-1\right)\right)\\
 & \hfill\qquad+\Lambda^{[1]}e^{2\beta_{0}^{[2]}}M_{R^{[1]}}^{\prime}\left(\Lambda^{[1]}\left(e^{2\beta_{0}^{[2]}}-1\right)\right)\Big\}\\
 & \hfill\qquad-\left\{ \Lambda^{[1]}e^{\beta_{0}^{[2]}}M_{R^{[1]}}^{\prime}\left(\Lambda^{[1]}\left(e^{\beta_{0}^{[2]}}-1\right)\right)\right\} ^{2}\Bigg]\\
 & =\Lambda^{[1]}\left(\Lambda^{[2]}\right)^{2}\left(1+b^{[2]}\right)e^{2\beta_{0}^{[2]}}\Bigg[\left(1+\psi^{[2]}\right)M_{R^{[1]}}^{\prime}\left(\zeta_{2}\right)+\Lambda^{[1]}e^{2\beta_{0}^{[2]}}M_{R^{[1]}}^{\prime\prime}\left(\zeta_{2}\right)\Bigg]\\
 & \quad-\bigg\{\Lambda^{[1]}\Lambda^{[2]}e^{\beta_{0}^{[2]}}M_{R^{[1]}}^{\prime}\left(\zeta_{1}\right)\bigg\}^{2},
\end{aligned}
\]
where the second last equality comes from Lemma \ref{dan.lem.30},
respectively.

The covariance of aggregate severities conditional on the priori premium
for $t_{1}\neq t_{2}$ is
\[
\begin{aligned}
 &\cov\left[{S_{t_{1}},S_{t_{2}}\bigg\vert\boldsymbol{\Lambda}}\right] \\
& =\cov\left[\E\left[{S_{t_{1}}\big\vert N_{t_{1}},\boldsymbol{\Lambda},\boldsymbol{R}}\right],\E\left[{S_{t_{2}}\big\vert N_{t_{2}},\boldsymbol{\Lambda},\boldsymbol{R}}\right]\Big\vert\boldsymbol{\Lambda}\right]+\E\left[{\cov{\left[S_{t_{1}},S_{t_{2}}\big\vert N_{t_{1}},N_{t_{2}},\boldsymbol{\Lambda},\boldsymbol{R}\right]}\Big\vert\boldsymbol{\Lambda}}\right]\\
 & =\cov\left[{\Lambda^{[2]}R^{[2]}N_{t_{1}}\exp\left(\beta_{0}^{[2]}N_{t_{1}}\right),\Lambda^{[2]}R^{[2]}N_{t_{2}}\exp\left(\beta_{0}^{[2]}N_{t_{2}}\right)\Big\vert\boldsymbol{\Lambda}}\right]\\
 & =\left(\Lambda^{[2]}\right)^{2}\Bigg\{\left(1+b^{[2]}\right)\E\left[{N_{t_{1}}N_{t_{2}}\exp\left(\beta_{0}^{[2]}\left(N_{t_{1}}+N_{t_{2}}\right)\right)\Big\vert\boldsymbol{\Lambda}}\right]\\
 & \hfill-\E\left[{N_{t_{1}}\exp\left(\beta_{0}^{[2]}N_{t_{1}}\right)\Big\vert\boldsymbol{\Lambda}}\right]
 \E\left[N_{t_{2}}\exp\left(\beta_{0}^{[2]}N_{t_{2}}\right)\Big\vert\boldsymbol{\Lambda}\right]\Bigg\}\\
 & =\left(\Lambda^{[2]}\right)^{2}\Bigg[\left(1+b^{[2]}\right)\left(\Lambda^{[1]}\right)^{2}e^{2{\beta_{0}^{[2]}}}M_{R^{[1]}}^{\prime\prime}\left(2\Lambda^{[1]}\left(e^{\beta_{0}^{[2]}}-1\right)\right)-\left\{ \Lambda^{[1]}e^{\beta_{0}^{[2]}}M_{R^{[1]}}^{\prime}\left(\Lambda^{[1]}\left(e^{\beta_{0}^{[2]}}-1\right)\right)\right\} ^{2}\Bigg]\\
 & =\left(\Lambda^{[1]}\Lambda^{[2]}\right)^{2}e^{2\beta_{0}^{[2]}}\Bigg[\left(1+b^{[2]}\right)M_{R^{[1]}}^{\prime\prime}\left(2\zeta_{1}\right)-\left\{ M_{R^{[1]}}^{\prime}\left(\zeta_{1}\right)\right\} ^{2}\Bigg],
\end{aligned}
\]
where the second last equality comes from Lemma \ref{dan.lem.30}.

The covariance among the frequencies conditional on the priori premium
for $t_{1}\neq t_{2}$ is
\[
\begin{aligned}\cov\left[{N_{t_{1}},N_{t_{2}}\big\vert\boldsymbol{\Lambda}}\right] & =\E\left[{\cov\left[{N_{t_{1}},N_{t_{2}}\big\vert\boldsymbol{\Lambda},\boldsymbol{R}}\right]\big\vert\boldsymbol{\Lambda}}\right]+\cov{\left[\E{\left[N_{t_{1}}\big\vert\boldsymbol{\Lambda},\boldsymbol{R}\right]},\E\left[{N_{t_{2}}\big\vert\boldsymbol{\Lambda},\boldsymbol{R}}\right]\big\vert\boldsymbol{\Lambda}\right]}\\
 & =\left(\Lambda^{[1]}\right)^{2}\Var\left[{R^{[1]}}\right].
\end{aligned}
\]

Note that the all last equalities of following proofs comes from Lemma
\ref{dan.lem.30}. The variance of the individual severities conditional
on the priori premium is
\[
\begin{aligned}\\
\Var\left[{Y_{t,j}\big\vert\boldsymbol{\Lambda}}\right] & =\E\left[{\Var\left[{Y_{t,j}\big\vert\boldsymbol{\Lambda},\boldsymbol{R}}\right]\Big\vert\boldsymbol{\Lambda}}\right]+\Var\left[{\E{\left[Y_{t,j}\big\vert\boldsymbol{\Lambda},\boldsymbol{R}\right]}\Big\vert\boldsymbol{\Lambda}}\right]\\
 & =\left(1+\psi^{[2]}\right)\E\left[{\left(\Lambda^{[2]}R^{[2]}\right)^{2}\exp\left(2\beta_{0}^{[2]}N_{t}\right)\Big\vert\boldsymbol{\Lambda}}\right]-\left\{ \E{\left[\Lambda^{[2]}R^{[2]}\exp\left(\beta_{0}^{[2]}N_{t}\right)\Big\vert\boldsymbol{\Lambda}\right]}\right\} ^{2}\\
 & =\left(\Lambda^{[2]}\right)^{2}\bigg[\left(1+b^{[2]}\right)\left(1+\psi^{[2]}\right)\,M_{R^{[1]}}\left(\zeta_{2}\right)-\left\{ M_{R^{[1]}}\left(\zeta_{1}\right)\right\} ^{2}\bigg].
\end{aligned}
\]

The covariances among the individual severities conditional on the
priori premium are

\[
\begin{aligned}\\
\cov\left[{Y_{t,j_{1}},Y_{t,j_{2}}\Big\vert\boldsymbol{\Lambda}}\right] & =\E{\left[\cov{\left[Y_{t,j_{1}},Y_{t,j_{2}}\big\vert\boldsymbol{\Lambda},\boldsymbol{R}\right]}\Big\vert\boldsymbol{\Lambda}\right]}+\cov\left[{\E{\left[Y_{t,j_{1}}\big\vert\boldsymbol{\Lambda},\boldsymbol{R}\right]},\E{\left[Y_{t,j_{2}}\big\vert\boldsymbol{\Lambda},\boldsymbol{R}\right]}\Big\vert\boldsymbol{\Lambda}}\right]\\
 & =\Var{\left[\Lambda^{[2]}R^{[2]}\exp\left(\beta_{0}^{[2]}N_{t}\right)\Big\vert\boldsymbol{\Lambda}\right]}\\
 & =\E\left[{\left(\Lambda^{[2]}R^{[2]}\right)^{2}\exp\left(2\beta_{0}^{[2]}N_{t}\right)\Big\vert\boldsymbol{\Lambda}}\right]-\left\{ \E{\left[\Lambda^{[2]}R^{[2]}\exp\left(\beta_{0}^{[2]}N_{t}\right)\Big\vert\boldsymbol{\Lambda}\right]}\right\} ^{2}\\
 & =\left(\Lambda^{[2]}\right)^{2}\bigg[\left(1+b^{[2]}\right)M_{R^{[1]}}\left(\zeta_{2}\right)-\left\{ M_{R^{[1]}}\left(\zeta_{1}\right)\right\} ^{2}\bigg]
\end{aligned}
\]
and, for $t_{1}\neq t_{2}$,

\[
\begin{aligned}\\
\cov\left[{Y_{t_{1},j_{1}},Y_{t_{2},j_{2}}\bigg\vert\boldsymbol{\Lambda}}\right] & =\E\left[{\cov{\left[Y_{t_{1},j_{1}},Y_{t_{2},j_{2}}\big\vert\boldsymbol{\Lambda},\boldsymbol{R}\right]}\Big\vert\boldsymbol{\Lambda}}\right]+\cov{\left[\E{\left[Y_{t_{1},j_{1}}\big\vert\boldsymbol{\Lambda},\boldsymbol{R}\right]},\E{\left[Y_{t_{2},j_{2}}\big\vert\boldsymbol{\Lambda},\boldsymbol{R}\right]}\Big\vert\boldsymbol{\Lambda}\right]}\\
 & =\cov{\left[\Lambda^{[2]}R^{[2]}\exp\left(\beta_{0}^{[2]}N_{t_{1}}\right),\Lambda^{[2]}R^{[2]}\exp\left(\beta_{0}^{[2]}N_{t_{2}}\right)\Big\vert\boldsymbol{\Lambda}\right]}\\
 & =\left(\Lambda^{[2]}\right)^{2}\Bigg\{\left(1+b^{[2]}\right)\E{\left[\exp\left(\beta_{0}^{[2]}\left(N_{t_{1}}+N_{t_{2}}\right)\right)\Big\vert\boldsymbol{\Lambda}\right]}\\
 & \hfill-\E{\left[\exp\left(\beta_{0}^{[2]}N_{t_{1}}\right)\Big\vert\boldsymbol{\Lambda}\right]}\E{\left[\exp\left(\beta_{0}^{[2]}N_{t_{2}}\right)\Big\vert\boldsymbol{\Lambda}\right]}\Bigg\}\\
 & =\left(\Lambda^{[2]}\right)^{2}\bigg[\left(1+b^{[2]}\right)M_{R^{[1]}}\left(2\Lambda^{[1]}\left(e^{\beta_{0}^{[2]}}-1\right)\right)-\left\{ M_{R^{[1]}}\left(\Lambda^{[1]}\left(e^{\beta_{0}^{[2]}}-1\right)\right)\right\} ^{2}\bigg].
\end{aligned}
\]

The covariances among the individual severities conditional on the
priori premium are
\[
\begin{aligned}\\
\cov{\left[N_{t},Y_{t,j}\Big\vert\boldsymbol{\Lambda}\right]} & =\E\left[{\E{\left[N_{t}\,Y_{t,j_{1}}\big\vert N_{t},\boldsymbol{\Lambda},\boldsymbol{R}\right]}\Big\vert\boldsymbol{\Lambda}}\right]-\E\left[{\E{\left[N_{t}\big\vert\boldsymbol{\Lambda},\boldsymbol{R}\right]}\Big\vert\boldsymbol{\Lambda}}\right]\E\left[{\E{\left[Y_{t,j_{1}}\big\vert\boldsymbol{\Lambda},\boldsymbol{R}\right]}\Big\vert\boldsymbol{\Lambda}}\right]\\
 & =\Lambda^{[2]}\E{\left[R^{[2]}N_{t}\exp\left(\beta_{0}^{[2]}N_{t}\right)\Big\vert\boldsymbol{\Lambda}\right]}-\Lambda^{[1]}\Lambda^{[2]}\E{\left[R^{[2]}\exp\left(\beta_{0}^{[2]}N_{t}\right)\Big\vert\boldsymbol{\Lambda}\right]}\\
 & =\Lambda^{[1]}\Lambda^{[2]}\bigg[e^{\beta_{0}^{[2]}}M_{R^{[1]}}^{\prime}\left(\zeta_{1}\right)-M_{R^{[1]}}\left(\zeta_{1}\right)\bigg]
\end{aligned}
\]
and, for $t_{1}\neq t_{2}$,

\[
\begin{aligned}\cov{\left[N_{t_{1}},Y_{t_{2},j}\bigg\vert\boldsymbol{\Lambda}\right]} & =\E\left[{\E{\left[N_{t_{1}}\,Y_{t_{2},j}\big\vert N_{t_{1}},\boldsymbol{\Lambda},\boldsymbol{R}\right]}\Big\vert\boldsymbol{\Lambda}}\right]-\E\left[\E{\left[N_{t_{1}}\big\vert\boldsymbol{\Lambda},\boldsymbol{R}\right]}\Big\vert\boldsymbol{\Lambda}\right]\E{\left[\E{\left[Y_{t_{2},j}\big\vert\boldsymbol{\Lambda},\boldsymbol{R}\right]}\Big\vert\boldsymbol{\Lambda}\right]}\\
 & =\Lambda^{[1]}\Lambda^{[2]}\E{\left[R^{[2]}\exp\left(\beta_{0}^{[2]}N_{t_{2}}\right)\Big\vert\boldsymbol{\Lambda}\right]}-\Lambda^{[1]}\Lambda^{[2]}\E{\left[R^{[2]}\exp\left(\beta_{0}^{[2]}N_{t_{2}}\right)\Big\vert\boldsymbol{\Lambda}\right]}\\
 & =0.\\
\\
\end{aligned}
\]
\end{proof}

Finally, we provide the MSE formulas for two B\"uhlmann methods in Section
\ref{sec:Two-B=0000FChlmann-Premiums}.

\begin{proposition}\label{prop.11} Under the settings in Model \ref{mod2},
we have
\[
\begin{aligned} & {\rm MSE}_{1}\left(\boldsymbol{\Lambda},t\right)\\
 & =\left(\Lambda^{[1]}\Lambda^{[2]}\right)^{2}\left(1+b^{[2]}\right)e^{2\beta_{0}^{[2]}}M_{R^{[1]}}^{\prime\prime}\left(2\zeta_{1}\right)+\widehat{\alpha}_{0}^{2}+t\widehat{\alpha}_{1}^{2}
 \left[\Var\left[S_{1}|\boldsymbol{\Lambda}\right]+\left(u_{1}\left(\boldsymbol{\Lambda}\right)\right)^{2}\right]\\
 & \quad\quad\quad\quad\quad\quad\quad+2t\widehat{\alpha}_{0}\widehat{\alpha}_{1}u_{1}\left(\boldsymbol{\Lambda}\right)+
 t(t-1)\widehat{\alpha}_{1}^{2}\left[\cov\left[S_{1},S_{2}|\boldsymbol{\Lambda}\right]+\left(u_{1}\left(\boldsymbol{\Lambda}\right)\right)^{2}\right]-2\widehat{\alpha}_{0}u_{1}\left(\boldsymbol{\Lambda}\right)\\
 & \quad\quad\quad\quad\quad\quad\quad\quad\quad\quad\quad\quad\quad\quad\quad\quad\quad\quad\quad\quad\quad\quad-2\widehat{\alpha}_{1}t\left(\Lambda^{[1]}\Lambda^{[2]}\right)^{2}e^{2\beta_{0}^{[2]}}\left(1+b^{[2]}\right)M_{R^{[1]}}^{\prime\prime}\left(2\zeta_{1}\right)
\end{aligned}
\]
and
\[
\begin{aligned} & {\rm MSE}_{2}\left(\boldsymbol{\Lambda},t\right)\\
 & =\left(\Lambda^{[1]}\Lambda^{[2]}\right)^{2}\left(1+b^{[2]}\right)e^{2\beta_{0}^{[2]}}M_{R^{[1]}}^{\prime\prime}\left(2\zeta_{1}\right)+\left(\widehat{\alpha}_{0}^{*}\right)^{2}\\
 & \quad+t\left(\widehat{\alpha}_{1}^{*}\right)^{2}\Lambda^{[1]}\left(\Lambda^{[2]}\right)^{2}e^{2\beta_{0}^{[2]}}\left[\Lambda^{[1]}e^{2\beta_{0}^{[2]}}M_{R^{[1]}}^{\prime\prime}\left(\zeta_{2}\right)+M_{R^{[1]}}^{\prime}\left(\zeta_{2}\right)\right]+2t\widehat{\alpha}_{0}^{*}\widehat{\alpha}_{1}^{*}u_{2}\left(\boldsymbol{\Lambda}\right)
\\
 & \quad\quad\quad\quad+t(t-1)\left(\widehat{\alpha}_{1}^{*}\Lambda^{[1]}\Lambda^{[2]}e^{\beta_{0}^{[2]}}\right)^{2}M_{R^{[1]}}^{\prime\prime}\left(2\zeta_{1}\right)-2\widehat{\alpha}_{0}^{*}u_{2}\left(\boldsymbol{\Lambda}\right),%
\end{aligned}
\]
where $\widehat{\alpha}_{j}$ and $\widehat{\alpha}_{j}^{*}$ for
$j=0,1,\cdots,t$ are defined in \eqref{alpha_agg} and  respectively.

\end{proposition}

\begin{proof}

First, ${\rm MSE}_{1}\left(\boldsymbol{\Lambda},t\right)$ can be
expressed as
\begin{equation}
\begin{aligned} & {\rm MSE}_{1}\left(\boldsymbol{\Lambda},t\right)\\
 & :=\E\left[\left(\E\left[S_{t+1}|\boldsymbol{R},\boldsymbol{\Lambda}\right]-{\rm Prem}_{1}\left(\boldsymbol{\Lambda},\mathcal{F}_{t}^{[{\rm agg}]}\right)\right)^{2}\Big|\boldsymbol{\Lambda}\right]\\
 & =\mathbb{E}\Bigg[\left(\E\left[S_{t+1}|\boldsymbol{R},\boldsymbol{\Lambda}\right]\right)^{2}+\left(\widehat{\alpha}_{0}^{2}+t\widehat{\alpha}_{1}^{2}S_{1}^{2}+2t\widehat{\alpha}_{0}\widehat{\alpha}_{1}S_{1}+t(t-1)\widehat{\alpha}_{1}^{2}S_{1}S_{2}\right)\\
 & \quad\quad\quad\quad\quad\quad\quad\quad\quad\quad\quad\quad\quad\quad-2\widehat{\alpha}_{0}
 \E\left[S_{t+1}|\boldsymbol{R},\boldsymbol{\Lambda}\right]-2\widehat{\alpha}_{1}tS_{1}
 \E\left[S_{t+1}|\boldsymbol{R},\boldsymbol{\Lambda}\right]\Big|\boldsymbol{\Lambda}\bigg],
\end{aligned}
\label{dan.20}
\end{equation}
where the second equality is just expansion of the square expression.
Finally, \eqref{dan.20} and the following equalities
\[
\begin{aligned}\E\left[2\widehat{\alpha}_{1}tS_{1}\E\left[{S_{t+1}|\boldsymbol{R},\boldsymbol{\Lambda}}\right]\Big
|\boldsymbol{\Lambda}\right]
=2\widehat{\alpha}_{1}t\E\left[\left(\E\left[S_{t}|\boldsymbol{R},\boldsymbol{\Lambda}\right]\right)^{2}\Big|\boldsymbol{\Lambda}\right]\end{aligned}
\]
and
\[
\begin{aligned}
\E\left[\left(\E\left[{S_{t+1}|\boldsymbol{R},\boldsymbol{\Lambda}}\right]\right)^{2}\Big|\boldsymbol{\Lambda}\right] & =\E\left[\left(\Lambda^{[1]}\Lambda^{[2]}R^{[1]}R^{[2]}e^{\beta_{0}^{[2]}}
\exp\left(\Lambda^{[1]}R^{[1]}\left(e^{\beta_{0}^{[2]}}-1\right)\right)\right)^{2}\Big|\boldsymbol{\Lambda}\right]\\
 & =\left(\Lambda^{[1]}\Lambda^{[2]}e^{\beta_{0}^{[2]}}\right)\left(1+b^{[2]}\right)M_{R^{[1]}}^{\prime\prime}\left(2\zeta_{1}\right)
\end{aligned}
\]
conclude the proof of the first part.

Now, ${\rm MSE}_{2}\left(\boldsymbol{\Lambda},t\right)$ can be expressed
as
\begin{equation}
\begin{aligned} & {\rm MSE}_{2}\left(\boldsymbol{\Lambda},t\right)\\
 & :=\E\left[\left(\E\left[{S_{t+1}|\boldsymbol{R},\boldsymbol{\Lambda}}\right] -{\rm Prem}_{2}\left(\boldsymbol{\Lambda},\mathcal{F}_{t}^{[{\rm freq}]}\right)\right)^{2}|\boldsymbol{\Lambda}\right]\\
 & =\mathbb{E}\bigg[\left(\E\left[S_{t+1}|\boldsymbol{R},\boldsymbol{\Lambda}\right]\right)^{2}+\left(\alpha_{0}^{*}\right)^{2}+t\left(\alpha_{1}^{*}\E{\left[S_{1}|N_{1},\boldsymbol{\Lambda}\right]}\right)^{2}+2t\alpha_{0}^{*}\alpha_{1}^{*}\E{\left[S_{1}|N_{1},\boldsymbol{\Lambda}\right]}\\
 & \quad\quad\quad\quad+t(t-1)\left(\alpha_{1}^{*}\right)^{2}\E{\left[S_{1}|N_{1},\boldsymbol{\Lambda}\right]}\E{\left[S_{2}|N_{2},\boldsymbol{\Lambda}\right]} -2\alpha_{1}t\E\left[S_{1}|N_{1},\boldsymbol{\Lambda}\right]\E\left[S_{t+1}|\boldsymbol{R},\boldsymbol{\Lambda}\right]\Big|\boldsymbol{\Lambda}\bigg],
\end{aligned}
\label{dan.18}
\end{equation}
where the second equality is just expansion of the square expression.
From Lemma \ref{dan.lem.30}, we also have
\begin{equation}
\begin{aligned}\E\left[{t\left(\alpha_{1}^{*}\E{\left[S_{1}|N_{1},\boldsymbol{\Lambda}\right]}\right)^{2}\Big|\boldsymbol{\Lambda}}\right] & =t\left(\alpha_{1}^{*}\Lambda^{[2]}\right)^{2}\E\left[{N_{t}^{2}\exp\left(2\beta_{0}^{[2]}N_{t}\right)\Big|\boldsymbol{\Lambda}}\right]\\
 & =t\left(\alpha_{1}^{*}\Lambda^{[2]}\right)^{2}\bigg[\left(\Lambda^{[1]}\right)^{2}e^{4\beta_{0}^{[2]}}M_{R^{[1]}}^{\prime\prime}\left(\zeta_{2}\right)+\Lambda^{[1]}e^{2\beta_{0}^{[2]}}M_{R^{[1]}}^{\prime}\left(\zeta_{2}\right)\bigg]
\end{aligned}
\label{dan.100}
\end{equation}
and
\begin{equation}
\begin{aligned}\E\left[{t(t-1)\left(\alpha_{1}^{*}\right)^{2}\E{\left[S_{1}|N_{1},\boldsymbol{\Lambda}\right]}\E{\left[S_{2}|N_{2},\boldsymbol{\Lambda}\right]}\Big|\boldsymbol{\Lambda}}\right] & =t(t-1)\left(\alpha_{1}^{*}\right)^{2}\E{\left[N_{1}N_{2}\exp\left(\beta_{0}^{[2]}\left(N_{1}+N_{2}\right)\right)\Big|\boldsymbol{\Lambda}\right]}\\
 & =t(t-1)\left(\alpha_{1}^{*}\Lambda^{[1]}\Lambda^{[2]}e^{\beta_{0}^{[2]}}\right)^{2}M_{R^{[1]}}^{\prime\prime}\left(2\zeta_{1}\right).
\end{aligned}
\label{dan.102}
\end{equation}
Furthermore, we have
\begin{equation}
\begin{aligned} & \E{\left[-2\alpha_{1}t\E{\left[S_{1}|N_{1},\boldsymbol{\Lambda}\right]}\E\left[{S_{t+1}|\boldsymbol{R},\boldsymbol{\Lambda}}\Big|\boldsymbol{\Lambda}\right]\right]}\\
 & =-2\alpha_{1}t\E{\left[\Lambda^{[2]}N_{1}\exp\left(\beta_{0}^{[2]}N_{1}\right)\Lambda^{[1]}\Lambda^{[2]}R^{[1]}R^{[2]}\exp\left(\Lambda^{[1]}R^{[1]}\left(e^{\beta_{0}^{[2]}}-1\right)\right)\Big|\boldsymbol{\Lambda}\right]}\\
 & =-2\alpha_{1}t\left(\Lambda^{[1]}\Lambda^{[2]}e^{\beta_{0}^{[2]}}\right)^{2}M_{R^{[1]}}^{\prime\prime}\left(2\zeta_{1}\right),
\end{aligned}
\label{dan.103}
\end{equation}
where the second equality is also from Lemma \ref{dan.lem.30}. Finally,
\eqref{dan.18}, \eqref{dan.100}, \eqref{dan.102}, and \eqref{dan.103}
conclude the proof of the second part.
\end{proof}

\section*{Appendix C: Tables}

\label{apdx.tabs}

\begin{table}[H]
\caption{(Data example) Estimation results under the frequency-severity Model
\ref{mod2} with dependence}
\vspace{-0.05in}
 \centering 
\begin{tabular}{lrrrrlrrrrlcc}
\hline
 &  &  & \multicolumn{2}{c}{95$\%$ CI} &  &  &  &  &  &  &  & \tabularnewline
\hline
parameter & Est & Std.dev & lower & upper &  &  &  &  &  &  &  & \tabularnewline
\hline
\multicolumn{3}{l}{\textbf{Frequency part}} &  &  &  &  &  &  &  &  &  & \tabularnewline
\quad{}Intercept & -1.884 & 0.292 & -2.442 & -1.294 & {*} &  &  &  &  &  &  & \tabularnewline
\quad{}City & 0.002 & 0.324 & -0.636 & 0.634 &  &  &  &  &  &  &  & \tabularnewline
\quad{}County & 1.279 & 0.317 & 0.644 & 1.883 & {*} &  &  &  &  &  &  & \tabularnewline
\quad{}School & -0.289 & 0.280 & -0.819 & 0.271 &  &  &  &  &  &  &  & \tabularnewline
\quad{}Town & -2.038 & 0.365 & -2.737 & -1.312 & {*} &  &  &  &  &  &  & \tabularnewline
\quad{}Village & -0.701 & 0.307 & -1.291 & -0.101 & {*} &  &  &  &  &  &  & \tabularnewline
\quad{}Coverage2 & 1.009 & 0.211 & 0.602 & 1.430 & {*} &  &  &  &  &  &  & \tabularnewline
\quad{}Coverage3 & 1.898 & 0.223 & 1.464 & 2.328 & {*} &  &  &  &  &  &  & \tabularnewline
\hline
\multicolumn{3}{l}{\textbf{Severity part}} &  &  &  &  &  &  &  &  &  & \tabularnewline
\quad{}Intercept & 8.394 & 0.366 & 7.712 & 9.140 & {*} &  &  &  &  &  &  & \tabularnewline
\quad{}City & -0.034 & 0.345 & -0.726 & 0.616 &  &  &  &  &  &  &  & \tabularnewline
\quad{}County & 0.527 & 0.333 & -0.126 & 1.169 &  &  &  &  &  &  &  & \tabularnewline
\quad{}School & -0.130 & 0.325 & -0.748 & 0.532 &  &  &  &  &  &  &  & \tabularnewline
\quad{}Town & 0.497 & 0.434 & -0.362 & 1.342 &  &  &  &  &  &  &  & \tabularnewline
\quad{}Village & 0.291 & 0.340 & -0.364 & 0.974 &  &  &  &  &  &  &  & \tabularnewline
\quad{}Coverage2 & 0.189 & 0.233 & -0.281 & 0.625 &  &  &  &  &  &  &  & \tabularnewline
\quad{}Coverage3 & 0.048 & 0.250 & -0.451 & 0.525 &  &  &  &  &  &  &  & \tabularnewline
\quad{}$\psi^{[2]}$ & 1.478 & 0.091 & 1.309 & 1.664 & {*} &  &  &  &  &  &  & \tabularnewline
\quad{}$\beta_{0}^{[2]}$ & -0.034 & 0.013 & -0.058 & -0.009 & {*} &  &  &  &  &  &  & \tabularnewline
\hline
\multicolumn{3}{l}{\textbf{Random effect part}} &  &  &  &  &  &  &  &  &  & \tabularnewline
\quad{}$b^{[1]}$ & 1.563 & 0.297 & 1.066 & 2.199 & {*} &  &  &  &  &  &  & \tabularnewline
\quad{}$b^{[2]}$ & 0.222 & 0.049 & 0.129 & 0.320 & {*} &  &  &  &  &  &  & \tabularnewline
\hline
\end{tabular}
\label{est.model2}
\end{table}

\bibliographystyle{plainnat}
\bibliography{Bib_Oh}

\end{document}